\documentclass[journal]{IEEEtran}
\usepackage{algorithmic}
\usepackage{algorithm}
\usepackage{amsfonts}
\usepackage{amssymb}
\usepackage{amsmath}
\usepackage{amsthm}
\usepackage{graphicx}
\usepackage{color}
\usepackage{multirow}
\usepackage{cite}
\usepackage{verbatim}
\usepackage{bm}
\usepackage{stfloats}

\usepackage{subcaption}
\usepackage[utf8]{inputenc}
\usepackage[english]{babel}
\usepackage{etex,etoolbox}
\usepackage{thmtools}
\usepackage{environ}
\usepackage{mathtools}
\usepackage{gensymb}
\newtheorem{theo}{Theorem}

\newtheorem{coro}{Corollary}

\def\dint{\displaystyle \int}%
\def\dsum{\mathop{\displaystyle \sum }}%
\begin{document}

\date{\thistime,\,\today}
\title{Multi-user Precoding and Channel Estimation for Hybrid Millimeter Wave Systems}
\author{Lou~Zhao, \IEEEmembership{Student Member,~IEEE}, Derrick~Wing~Kwan~Ng, \IEEEmembership{Member,~IEEE}, \\
and Jinhong Yuan, \IEEEmembership{Fellow,~IEEE}
\thanks{%
Part of this manuscript has been accepted for presentation at the IEEE ICC 2017 \cite{Zhao2017}.
L. Zhao, D.~W.~K.~Ng, and J. Yuan are with the School of Electrical
Engineering and Telecommunications, The University of New South Wales,
Sydney, Australia (Email: lou.zhao@unsw.edu.au; w.k.ng@unsw.edu.au;
j.yuan@unsw.edu.au). }}
\maketitle

\begin{abstract}
In this paper, we develop a low-complexity channel estimation for hybrid millimeter wave (mmWave) systems, where the number of radio frequency (RF) chains is much less than the number of antennas equipped at each transceiver.
The proposed mmWave channel estimation algorithm first exploits multiple frequency tones to estimate the strongest angle-of-arrivals (AoAs) at both base station (BS) and user sides for the design of analog beamforming matrices.
Then all the users transmit orthogonal pilot symbols to the BS along the directions of the estimated strongest AoAs in order to estimate the channel.
The estimated channel will be adopted to design the digital zero-forcing (ZF) precoder at the BS for the multi-user downlink transmission.
The proposed channel estimation algorithm is applicable to both non-sparse and sparse mmWave channel environments.
Furthermore, we derive a tight achievable rate upper bound of the digital ZF precoding with the proposed channel estimation algorithm scheme.
Our analytical and simulation results show that the proposed scheme obtains a considerable achievable rate of fully digital systems, where the number of RF chains equipped at each transceiver is equal to the number of antennas.
Besides, by taking into account the effect of various types of errors, i.e., random phase errors, transceiver analog beamforming errors, and equivalent channel estimation errors, we derive a closed-form approximation for the achievable rate of the considered scheme.
We illustrate the robustness of the proposed channel estimation and multi-user downlink precoding scheme against the system imperfection.

\end{abstract}
\renewcommand{\baselinestretch}{0.99}
\large\normalsize

\vspace*{-0mm}

\section{Introduction}
\vspace*{-0mm}

Higher data rates, larger bandwidth, and higher spectral efficiency are necessary for the fifth-generation (5G) wireless communication systems to support various emerging applications. The combination of millimeter wave (mmWave) communication \cite{AZhang2015,Kokshoorn2016,Elkashlan2014,Sohrabi2016} with massive multiple-input multiple-output (MIMO) \cite{Yang2015,Bogale2015,Deng2015,Marzetta2010,Yang2015bb,Dai2013} is considered as one of the promising candidate technologies for 5G communication systems with many potential opportunities for research \cite{Ng2017,Bogale2015,Rappaport2015}.

Communication in mmWave band (frequency ranges from $30$ GHz to $300$ GHz) was not widely applied to cellular systems due to the inherent high propagation path loss, low penetration coefficients, and high signal attenuation caused by raindrop absorption.
Recently, mmWave has attracted growing interests from both academia and industry.
In particular, it is considered as one of important technologies for dynamic micro-cell or pico-cell (IEEE 802.11ad) systems with small coverage since it provides a tremendous spectrum which is available throughout the world \cite{Heath2016a}.
On the other hand, massive MIMO systems, equipping a base station (BS) with hundreds of antennas, can be exploited to serve tens of users simultaneously to enhance the system spectral efficiency.
{More importantly, low computational complexity linear precoding schemes, such as maximum-ratio combining (MRC)/maximum-ratio transmission (MRT) and zero-forcing (ZF), can be deployed and to achieve high data rates due to a large amount of spatial degrees of freedom \cite{Alkhateeb2015,Dai2015,Marzetta2010} available in massive MIMO systems.
Though dirty paper coding (DPC) can pre-cancel known interference without power penalty, the non-linear algorithm may not be suitable for implementing in practical systems due to its high computational  complexity.}
Therefore, massive MIMO systems operating at mmWave band is expected to provide many potential advantages and exciting opportunities for future research\cite{Alkhateeb2015,Swindlehurst2014,Bjornson2016}.
In fact, there are plenty of implementation challenges for mmWave massive MIMO communication systems.
For example, the trade-offs between system performance, hardware complexity\footnote{The hardware includes power amplifier (PA), analog digital converter/digital analog converter (ADC/DAC), phase shifters, and antenna array.}, and energy consumption \cite{AZhang2015,Heath2016a} are still unclear.
From the literature, it is certain that conventional fully digital MIMO systems, in which each antenna connects with a dedicated radio frequency (RF) chain, are impractical for mmWave systems due to the prohibitively high cost, e.g. tremendous energy consumption of high resolution ADC/DACs and PAs.
Therefore, several mmWave hybrid systems were proposed as compromised solutions which strike a balance between hardware complexity and system performance \cite{Alkhateeb2015,Ni2016,Han2015,Ayach2014,Sohrabi2016,Gao2016,Bjornson2016}.
Specifically, the use of a large number of antennas, connected with only a small number of independent RF chains at transceivers, is adopted to exploit the large array gain to compensate the inherent high path loss in mmWave channels \cite{Rappaport2015,Hur2016}.
Yet, the hybrid system imposes a restriction on the number of RF chains which introduces a paradigm shift in the design of both resource allocation algorithms and transceiver signal processing.

Conventionally, pilot-aided channel estimation algorithms are widely adopted for fully digital time-division duplex (TDD) massive MIMO systems \cite{Marzetta2010,Shen2015} operating in sub-$6$ GHz frequency bands.
However, these algorithms cannot be  directly applied to hybrid mmWave systems as the number of RF chains is much smaller than the number of antennas.
In fact, for the channel estimation in hybrid mmWave systems, the strategies of allocating analog/digital beams to different users and estimating the equivalent baseband channels are still an open area of research \cite{Alkhateeb2015,Prelcic2016}.
In addition, a channel estimation algorithm designed for a specific type of hybrid system may not be applicable to other hybrid systems \cite{Alkhat2014}.
For instance, in the multi-user (MU) channel estimation with analog beam training approaches, open-loop beamforming (OLB) channel estimation is widely adopted \cite{Bjornson2016}.
However, OLB channel estimation is not suitable for large scale antenna arrays, as the required amount of feedback bits scale with the number of transmit antennas.
Recently, several improved mmWave channel estimation algorithms were proposed \cite{Kokshoorn2016,Alkhat2014,Alkhateeb2015}.
The overlapped beam patterns and rate adaptation channel estimation were investigated in \cite{Kokshoorn2016} to reduce the required training time for channel estimation.
Then, the improved limited feedback channel estimation was proposed \cite{Alkhateeb2015} to maximize the received signal power at each single user so as to reduce the required training and feedback overheads.
{However, explicit channel state information (CSI) feedback from users is still required for these channel estimation algorithms. In practice, CSI feedback may cause high complexity and extra signallings overhead.
In addition, there will be a system rate performance degradation due to the limited amount of the feedback and the limited resolution of CSI quantization.
Therefore, a low computational complexity mmWave channel estimation algorithm, which does not require explicit CSI feedback, is necessary to unlock the potential of hybrid mmWave systems.}

In the literature, most of the existing mmWave channel estimation algorithms leverage the sparsity of mmWave channels due to the extremely short wavelength of mmWave \cite{Heath2016a,Andrews2016}.
Generally, in suburban area or outdoor long distance propagation environment \cite{Hur2016}, the sparsity of mmWave channels can be well exploited.
In practical urban area (especially in the city center), the number of unexpected scattering clusters increases significantly and mmWave communication channels may not be necessarily sparse.
For instance, in the field measurements in Daejeon city, Korea, and the associated ray-tracing simulation \cite{Hur2016}, the angle of arrivals (AoAs) at the BS and the users were observed under the impact of non-negligible scattering clusters.
Hence, some preliminary works \cite{Ayach2014,Liang2014} have started to investigate the clustered channel model for modeling MIMO propagation at mmWave carrier frequencies \cite{Buzzi2016b}.
In addition, existing mmWave channel estimation algorithms \cite{Kokshoorn2016,Alkhat2014,Alkhateeb2015}, which are designed based on the assumption of channel sparsity, may not be suitable for non-sparse mmWave channels.
Indeed, the scattering clusters of mmWave channels due to macro-objects or backscattering from the objects, have a significant impact on system performance and cannot be neglected in the system design.
Therefore, there is an emerging need for a channel estimation algorithm which is applicable to both non-sparse and sparse mmWave channels.

Motivated by aforementioned discussions, we consider a MU hybrid mmWave system.
In particular, we propose and detail a novel non-feedback non-iterative channel estimation algorithm which is applicable to both non-sparse and sparse mmWave channels.
Also, we analyze the achievable rate performance of the mmWave system using digital ZF precoding based on the estimated channel information.
Furthermore, we analyze the performance degradation under some practical hardware imperfections, such as random phase errors, RF transceiver beamforming errors, and channel estimation errors.
Our main contributions are summarized as follows:

\begin{itemize}

\item We propose a three-step MU channel estimation scheme for mmWave channels.
In the first two steps, we estimate the strongest AoAs at both the BS and the users sides instead of estimating the combination of multiple AoAs.
The estimated strongest AoAs will be exploited for the design of analog beamforming matrices at the BS and users.
In the third step, all the users transmit orthogonal pilot symbols to the BS along the beamforming paths of the strongest AoA directions to facilitate the equivalent channel estimation, which will be exploited to design the digital ZF precoder at the BS for the downlink transmission.
Our proposed hybrid scheme can suppress the downlink MU interference effectively via its analog beamforming and digital precoder.
Firstly, the proposed analog beamforming allow signal transmission and reception along the strongest AoA direction, which reduces the interference outside the strongest AoA directions and utilizes the transmission power more efficiently.
Secondly, the digital ZF precoder can suppress the MU interference within the strongest AoA directions.

\item We analyze the achievable rate performance of the proposed scheme based on the estimated equivalent channel CSI, analog beamforming matrices, and digital ZF precoding.
While assuming the equivalent CSI is perfectly known at the BS, we derive a tight performance upper bound on the achievable rate of our proposed scheme.
Also, we quantify the performance gap between the proposed hybrid scheme and the fully digital system in terms of achievable rate per user.
It is interesting to note that the performance gap is determined by the ratio between the power of the strongest AoA component and the power of the scattering component, Rician K-factor $\upsilon$.
The performance gap of the average achievable rate per user between the hybrid system and the fully digital system is only $\left\vert\log_2 \left( \frac{\upsilon }{\upsilon +1}\right)\right\vert$ bits/s/Hz in the large numbers of antennas regime.

\item We further analyze the system performance degradation and derive the closed-form approximation of achievable rate under various types of errors, i.e., random phase errors, transceiver analog beamforming errors, and equivalent channel estimation errors, in the high receiver signal-to-noise ratio (SNR) and the large numbers of antennas regimes.
Interestingly, our results confirm that the impact of phase errors and transceiver beamforming errors will not cause a performance ceiling in terms of achievable rate.
Besides, the performance gap in terms of the achievable rate between the system under phase errors and transceiver beamforming errors and the system with perfect hardware is approximated.

\end{itemize}

Notation: $\mathrm{E}_{\mathrm{h}}(\cdot )$ denote statistical expectation operation with respect to random variable $h$, $\mathbb{C}^{M\times N}$ denotes the space of all $M\times N$ matrices with complex entries; $(\cdot )^{-1}$ denotes inverse operation; $(\cdot )^{H}$ denotes Hermitian transpose; $(\cdot )^{\ast }$ denotes complex conjugate; $(\cdot)^{T}$ denotes transpose; $|\cdot |$ denote the absolute value of a complex scalar; $\mathrm{tr}(\cdot )$ denotes trace operation; $\|\cdot\|_{\mathrm{F}}$ denotes the Frobenius norm of matrix; $\lambda _{i}(\cdot )$ denotes the $i$-th maximum eigenvalue of a matrix; $\mathrm{diag}\left\{a\right\} $ is a diagonal matrix with the entries $a$ on its diagonal.
The distribution of a circularly symmetric complex Gaussian (CSCG) random vector with a mean vector $\mathbf{x}$ and a covariance matrix ${\sigma}^{2}\mathbf{I}$  is denoted by ${\cal CN}(\mathbf{x},{\sigma}^{2}\mathbf{I})$, and $\sim$ means ``distributed
as".  

\vspace*{-1mm}
\section{System Model}

We consider a MU hybrid mmWave system which consists of one base station (BS) and $N$ users in a single cell, as shown in Figure \ref{fig:hybrid}.
{Generally, there are two kind of hybrid structures which are widely adopted by researchers \cite{AZhang2015,Alkhateeb2015}: the full access hybrid architecture and the subarray hybrid architecture. The full access hybrid architecture, where each RF chain is connected to all the antennas, can provide a higher array gain and a narrower beam width than that of the subarray hybrid architecture, where each RF chain is connected to a part of the antennas. In this paper, we adopt the full access hybrid architecture since it offers higher flexibility in the design of channel estimation algorithm.}
We assume that the BS is equipped with $M\geq 1$ antennas and $N_{\mathrm{RF}}$  radio frequency (RF) chains to serve the $N$ users.
{Besides, each user is equipped with $P$ antennas and a single RF chain. We also assume that $M\geqslant N_{\mathrm{RF}}\geqslant N$.}
In the following sections, we set $N = N_{\mathrm{RF}}$ to simplify the analysis\footnote{We note that our proposed channel estimation scheme, precoding scheme, and analysis can be generalized to the case of $N_{\mathrm{RF}}\geq N$,  at the expense of a more involved notation.}.
Each RF chain at the BS can access to all the antennas by using $M$ phase shifters, as shown in Figure \ref{fig:RF_chain}.
At each BS, the number of phase shifters is $M\times N_{\mathrm{RF}}$.
Due to significant propagation attenuation at mmWave frequency, the system is dedicated to cover a small area, e.g. cell radius is $\sim 150$ m.
We assume that the users and the BS are fully synchronized and time division duplex (TDD) is adopted to facilitate uplink and downlink communications \cite{Marzetta2010}.
\begin{figure}[t]\vspace{-0mm}
\centering \includegraphics[width=3.5in]{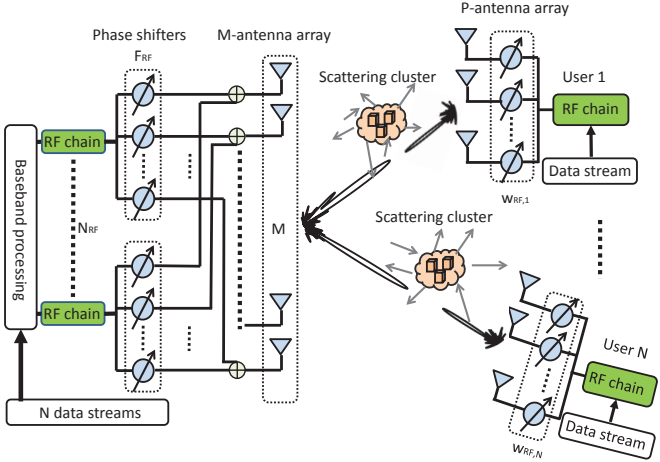}
\vspace{-5mm}
\caption{A mmWave massive MIMO communication system with a hybrid system of
transceivers.}
\label{fig:hybrid}\vspace{-4mm}
\end{figure}

In previous work \cite{Alkhateeb2015}, mmWave channels were assumed to have sparse propagation paths between the BS and the users.
{Yet, in recent field tests, especially in the urban microcell environments, both a strong line-of-sight (LOS) component and non-negligible scattering components may exist in mmWave propagation channels \cite{Buzzi2016b,Rappaport2015,Hur2016}.
Therefore, for the urban short-distance propagation environment, mmWave channels are more suitable to be modeled by non-sparse Rician fading and with a large Rician K-factor \cite{Hur2016,Rappaport2015,Al-Daher2012}.
On the other hand, for the suburban long-distance propagation environment, mmWave channels can be modeled by a sparse channel model. The reason is that scattering components will vanish during the long-distance propagation because of the high reflection loss and large propagation path loss.}
{In addition, the blockage of LOS component is critical for mmWave systems and widely considered in previous works \cite{Bai2015b,Andrews2016}. In \cite{Bai2015b}, the authors proved that the optimal cell size to achieve the maximum SINR scales with the average size of the area that is LOS to a user. Therefore, with a properly designed radius of cells, one should expect the existence of at least one LOS component from the BSs to any user. When the LOS component from a certain BS to a user is blocked, the user can still exploit other existing LOS components from other BSs for channel estimation and data transmission.}
\begin{figure}[t]
\centering\vspace*{-0mm}
\includegraphics[height=1.4in]{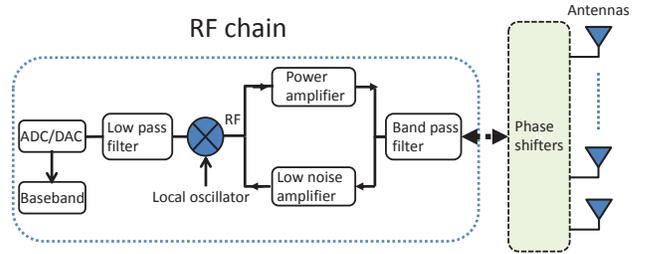}
\vspace{-1mm}
\caption{A block diagram of a RF chain for an antenna array.}
\label{fig:RF_chain}\vspace{-5mm}
\end{figure}
Let $\mathbf{H}_{k}\in\mathbb{C}^{M\times P}$ be the uplink channel matrix between user $k$ and the BS in the cell.
We assume that $\mathbf{H}_{k}$ is a slow time-varying block Rician fading channel, i.e., the channel is constant in a block but varies slowly from one block to another.
Then, in this paper, we assume that the channel matrix $\mathbf{H}_k$ can be decomposed into a deterministic LOS channel matrix $\mathbf{H}_{\mathrm{L},k}\in\mathbb{C}^{M\times P}$ and a scattered channel matrix $\mathbf{H}_{\mathrm{S,}k}\in\mathbb{C}^{M\times P}$ \cite{Buzzi2016b}, i.e., \vspace*{-0mm}
\begin{equation}\label{eqn:LOS_channel}
\vspace*{-1mm}
\mathbf{H}_{k}=\underset{\mathrm{LOS}\text{\ }\mathrm{component}}{\underbrace{\mathbf{H}_{\mathrm{L,}k}\mathbf{G}_{\mathrm{L,}k}}}+\underset{\mathrm{Scattering}\text{\ }\mathrm{component}}{\underbrace{\mathbf{H}_{\mathrm{S,}k}\mathbf{G}_{\mathrm{S,}k}}},\vspace*{-1mm}
\end{equation}\vspace*{-0mm}%
where $\mathbf{G}_{\mathrm{L,}k}\in\mathbb{C}^{P\times P}$ and $\mathbf{G}_{\mathrm{S},k}\in\mathbb{C}^{P\times P}$ are diagonal matrices with entries
\begin{equation}
\vspace*{-1mm}
\mathbf{G}_{\mathrm{L,}k}=\mathrm{diag}\left\{ \sqrt{\frac{\upsilon _{k}}{\upsilon _{k}+1}}\right\} \text{ and } \mathbf{G}_{\mathrm{S,}k}=\mathrm{diag}\left\{ \sqrt{\frac{1}{\upsilon _{k}+1}}\right\},\vspace*{-1mm}
\end{equation}%
respectively, and $\upsilon _{k}>0$ is the Rician K-factor of user $k$.
{{Besides, Equations (1) and (2) are the generalization of mmWave channel models, which capture both the scattered and non-scattered components.}
In general, we can adopt different array structures, e.g. uniform linear array (ULA) and uniform panel array (UPA) for both the BS and the users.
Here, we adopt the ULA as it is commonly implemented in practice \cite{Alkhateeb2015}.
We assume that all the users are separated by hundreds of wavelengths or more \cite{Marzetta2010}.
Thus, we can express the deterministic LOS channel matrix $\mathbf{H}_{\mathrm{L},k}$ of user $k$ as \cite{book:wireless_comm}\vspace*{-0mm}
\begin{equation}
\mathbf{H}_{\mathrm{L,}k}=\mathbf{h}_{\mathrm{L,}k}^{\mathrm{BS}}\mathbf{h}_{%
\mathrm{L,}k}^{H},
\end{equation}\vspace*{-0mm}%
where $\mathbf{h}_{\mathrm{L},k}^{\mathrm{BS}}$ $\in\mathbb{C}^{M\times 1}$ and $\mathbf{h}_{\mathrm{L,}k}$ $\in\mathbb{C}^{P\times 1}$ are the antenna array response vectors of the BS and user $k$ respectively.

In particular, $\mathbf{h}_{\mathrm{L,}k}^{\mathrm{BS}}$ and $\mathbf{h}_{\mathrm{L,}k}$ can be expressed as \cite{book:wireless_comm,Trees2002}\vspace*{-1mm}
\begin{align}
\vspace*{-1mm}
\mathbf{h}_{\mathrm{L},k}^{\mathrm{BS}}& =\left[
\begin{array}{ccc}
1, & \ldots
, & \text{ }e^{-j2\pi \left( M-1\right) \tfrac{d}{\lambda }\cos \left(
\theta _{k}\right) }%
\end{array}%
\right] ^{T} \text{and} \\
\mathbf{h}_{\mathrm{L},k}& =\left[
\begin{array}{ccc}
1, & \ldots ,
& \text{ }e^{-j2\pi \left( M-1\right) \tfrac{d}{\lambda }\cos \left( \phi
_{k}\right) }%
\end{array}%
\right] ^{T},\vspace*{-1mm}
\end{align}\vspace*{-0mm}
respectively, where $d$ is the distance between the neighboring antennas and $\lambda $ is the wavelength of the carrier frequency.
Variables $\theta _{k}\in \left[ 0,+\pi \right]$ and $\phi _{k}\in \left[ 0,+\pi \right] $ are the angles of incidence of the LOS path at antenna arrays of the BS and user $k$, respectively.
For convenience, we set $d=\frac{\lambda }{2}$ for the rest of the paper which is an assumption commonly adopted in the literature \cite{Trees2002,book:wireless_comm}.
Without loss of generality, we assume that the scattering component $\mathbf{H}_{\mathrm{S,}k}$ consists $N_{\mathrm{cl}}$ clusters and each cluster contributes $N_{\mathrm{l},i}$ propagation paths \cite{Buzzi2016b}, which can be expressed as\vspace*{-1mm}
\begin{align}
\vspace*{-0mm}
\mathbf{H}_{\mathrm{S,}k}&=\sqrt{\tfrac{1}{{\sum }_{i=1}^{N_{\mathrm{cl}}}{N_{\mathrm{l},i}}}}\overset{N_{\mathrm{cl}}}{\underset{i=1}{\dsum }}\overset{N_{\mathrm{l},i}}{\underset{l=1}{\dsum }}{\alpha _{i,l}}\mathbf{h}_{i,l}^{\mathrm{BS}}\mathbf{h}_{k,i,l}^{H}\notag\\
&= \left[  \begin{array}{ccccc} \mathbf{h}_{\mathrm{S},1}, &\ldots,& \mathbf{h}_{\mathrm{S},k}, &\ldots,& \mathbf{h}_{\mathrm{S},P} \end{array}\right],\vspace*{-1mm}
\end{align}\vspace*{-0mm}%
where $\mathbf{h}_{i,l}^{\mathrm{BS}}\in\mathbb{C}^{M\times 1}$ and $\mathbf{h}_{k,i,l}\in\mathbb{C}^{P\times 1}$ are the antenna array response vectors of the BS and user $k$ associated to the $\left( i,l\right) $-th propagation path, respectively.
Here $\alpha _{i,l}\sim \mathcal{CN}\left( 0,1\right) $ represents the path attenuation of the $\left(i,l\right)$-th propagation path and $\mathbf{h}_{\mathrm{S},k}\in\mathbb{C}^{M\times 1}$ is the $k$-th column vector of $\mathbf{H}_{\mathrm{S},k}$.
With the increasing number of clusters, the path attenuation coefficients and the AoAs between the users and the BS become randomly distributed \cite{Buzzi2016b,Hur2016}.
Therefore, we model the entries of scattering component $\mathbf{H}_{\mathrm{S,}k}$ in a general manner as an independent and identically distributed (i.i.d.) random variable\footnote{To facilitate the study of the downlink hybrid precoding, we assume that perfect long-term power control is performed to compensate for path loss and shadowing at the desired users and equal power allocation among different data streams of the users\cite{Alkhateeb2015,Ni2016,Yang2013}. Thus, the entries of scattering component $\mathbf{H}_{\mathrm{S,}k}$ are modeled by i.i.d. random variables. } $\mathcal{CN}\left( 0,1\right)$.
\vspace*{-0mm}
\section{Proposed Channel Estimation for Hybrid System}
\begin{figure}[t]
\centering\vspace{-0mm}
\includegraphics[width=3.4in]{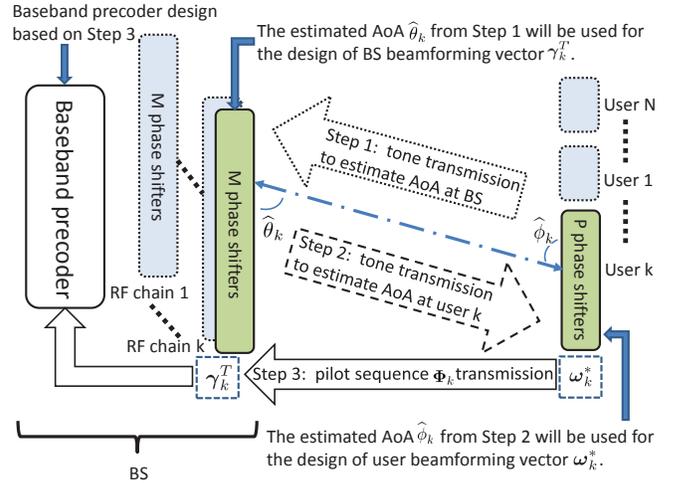}\vspace{-0mm}
\caption{An illustration of the proposed channel estimation algorithm for hybrid mmWave systems.}
\label{fig:CEI}\vspace*{-3mm}
\end{figure}

\begin{algorithm}
\caption{Channel Estimation Algorithm for Hybrid Systems}
\label{a1}
\begin{algorithmic} [1]
\REQUIRE Multiple single-carrier frequency tone signals $[f_1, \text{}\ldots, f_N]$, pilot sequences matrix $\mathbf{\Psi }$,
         the detection matrices $[\mathbf{\Gamma}_1, \text{}\ldots, \mathbf{\Gamma}_N] $ of AoA at the BS,
         and the detection matrices $[\mathbf{\Omega}_1, \text{}\ldots, \mathbf{\Omega}_N]$ of AoA at the users\\\vspace*{+2mm}
{{ STEP 1: \underline{Estimate the AoA at the BS and}}}\\
{{ \quad \quad \quad \quad \underline{design BS analog beamforming matrix}}}\vspace*{+1mm}
\STATE  The estimation of AoA at the BS: all the users transmit their unique frequency tones by using only one omni-directional antenna
\STATE  The BS calculates $r_{k,i}^{\mathrm{BS}}, \text{\ } i\in \{1,\ldots,J\}$, as shown in (\ref{AoA_D}) to estimate the uplink AoA of user $k$ and its corresponding beamforming vector:
{$\widetilde{\mathbf{\gamma}}_{k}=\underset{\forall \gamma_{k,i}, \text{\ }i\in \{1,\ldots,J\}}{\arg \max \left\vert r_{k,i}^{\mathrm{BS}}\right\vert }$}
\STATE The BS analog beamforming matrix is \\
$\mathbf{F}_{\mathrm{RF}}=\left[\begin{array}{ccc}\widetilde{\bm{\gamma}}_{1},\ldots, \widetilde{\bm{\gamma}}_{N}\end{array}\right]$\\\vspace*{+2mm}
{{ STEP 2: \underline{Estimate the AoA at the users and}}}\\
{{\quad \quad \quad \quad \underline{design user analog beamforming matrix}}}\vspace*{+1mm}
\STATE  The estimation of AoA at the users: the BS transmits frequency tones back to all the users using $\mathbf{F}_{\mathrm{RF}}$ as a transmit beamforming matrix
\STATE  Calculate $r _{k,i}^{\mathrm{UE}},\text{\ } i\in \{1,\ldots,J\}$, as shown in (\ref{AoA_D2}) to estimate the downlink AoA of user $k$ and its corresponding beamforming vector: $\widetilde{\mathbf{\omega }}_{k}^{\ast }=\underset{\forall \mathbf{\omega }_{k,i}, \text{\ }i\in \{1,\ldots,J\}}{\arg \max \left\vert \gamma
_{k,i}^{\mathrm{UE}}\right\vert }$
\STATE The users analog beamforming matrix is \\
$\mathbf{Q}_{\mathrm{RF}}=\left[
\begin{array}{ccc}
\widetilde{\mathbf{\omega }}_{1}^{\ast}, \ldots, \widetilde{\mathbf{\omega}}_{N}^{\ast }\end{array}\right] $\\ \vspace*{+2mm}
{{STEP 3: \underline{Estimate equivalent channel and}}}\\
{{ \quad \quad \quad \quad \underline{design digital ZF precoder}}}\vspace*{+1mm}
\STATE All the users transmit orthogonal pilot sequences by using $\mathbf{Q}_{\mathrm{RF}}$ as beamforming matrix and the BS uses $\mathbf{F}^{T}_{\mathrm{RF}}$ as beamforming matrix to receive pilot sequences
\STATE The BS obtains and calculates $\widehat{\mathbf{H}}_{\mathrm{eq}}^{T}$ as shown in (\ref{EHC_1}) $\widehat{\mathbf{H}}_{\mathrm{eq}}^{T}=\mathbf{\Psi }^{H}\left[\begin{array}{ccccc}\mathbf{s}_{1},\text{}\ldots ,\mathbf{s}_{N}\end{array}\right]$\\
\STATE The BS sets the baseband digital ZF precoder as \\ $\overline{\mathbf{W}}_{\mathrm{eq}}=\widehat{\mathbf{H}}_{\mathrm{eq}}^{\ast }(\widehat{\mathbf{H}}_{\mathrm{eq}}^{T}\widehat{\mathbf{H}}_{\mathrm{eq}}^{\ast })^{-1}$
\end{algorithmic}
\end{algorithm}
In this section, we propose and detail our mmWave channel estimation for hybrid mmWave systems.
In practice, the hybrid system imposes a fundamental challenge for mmWave channel estimation.
Unfortunately, the conventional pilot-aided channel estimation algorithm for fully digital systems, e.g. \cite{Kokshoorn2016,Alkhateeb2015}, is not applicable to the considered hybrid mmWave system.
The reasons are that the number of RF chains is much smaller than the number of antennas equipped at the BS and the transceiver beamforming matrix cannot be acquired.

To address this important issue, we propose a novel channel estimation algorithm, which contains three steps as shown in Figure \ref{fig:CEI} and Algorithm \ref{a1}.
{In the first and second steps, we introduce unique unmodulated frequency tones to estimate the strongest AoAs at the BS and user sides. The unique frequency tones and linear search algorithm are inspired by signal processing in monopulse passive electronically scanned array (PESA) radar and sonar systems \cite{book:wireless_comm}.}
These estimated strongest AoAs will be exploited to develop analog transmit and receive beamforming matrices at the BS and users.
In the third step, the users transmit orthogonal pilot symbols to the BS along the beamforming paths in order to estimate the equivalent channel via the strongest AoA directions.
Then, the estimated channel will be used for the design of BS digital baseband precoder for the downlink transmissions by exploiting the reciprocity between the uplink and downlink channels.

\vspace*{-1mm}
\subsection{Details of Proposed Channel Estimation}

\quad\emph{Step $1$}, \underline{Line 1 in Algorithm \ref{a1}:} Firstly, all the users transmit unique frequency tones to the desired BS in the uplink simultaneously.
For user $k$, an unique unmodulated frequency tone, $x_{k}=\cos \left( 2\pi f_{k}t\right), k\in \{1,\cdots ,N\}$, is transmitted from one of the omni-directional antennas in the antenna array to the BS.
Here, $f_{k}$ is the single carrier frequency and $t$ stands for time and $f_{k}\neq f_{j},  \forall k\neq j$.
For the AoA estimation, if the condition $\frac{f_{k}-f_{j}}{f_\mathrm{c}} < 10^{-4}, \text{ }\forall k\neq j$, is satisfied, the AoA estimation difference by using different tones is generally negligible \cite{Trees2002}, where $f_\mathrm{c}$ is the system carrier frequency.
The pass-band received signal of user $k$ at the BS, $\mathbf{y}_{k}^{\mathrm{BS}}$, is given by\vspace*{-1mm}
\begin{equation}
\vspace*{-1mm}
\mathbf{y}_{k}^{\mathrm{BS}}=\left( \sqrt{\dfrac{\upsilon _{k}}{%
\upsilon _{k}+1}}\mathbf{h}_{\mathrm{L},k}^{\mathrm{BS}}+\sqrt{\dfrac{1}{%
\upsilon _{k}+1}}\mathbf{h}_{\mathrm{S},k}\right) x_{k}+\mathbf{z}_{\mathrm{%
BS}}, \vspace*{-0mm}
\end{equation}\vspace*{-0mm}%
where $\mathbf{z}_{\mathrm{BS}}$ denotes the thermal noise at the antenna array of the BS, $\mathbf{z}_{\mathrm{BS}}\sim\mathcal{CN}\left( \mathbf{0},{\sigma_{\mathrm{BS}}^{2}}\mathbf{I} \right)$, and ${\sigma} _{\mathrm{BS}}^{2}$ is the noise variance at each antenna of the BS.
To facilitate the estimation of AoA, we perform a linear search in the angular domain ranged from $0^{\circ}$ to $180^{\circ}$ with an angle search step size of $\frac{180}{J}$.
{Therefore, the AoA detection matrix $\mathbf{\Gamma }_{k}\in\mathbb{C}^{M\times J}$, $\mathbf{\Gamma }_{k}= \left[ \begin{array}{ccc} \bm{\gamma }_{k,1},\ldots,\bm{\gamma }_{k,J} \end{array}\right]$, contains $J$ column vectors. In general, the typical value of the minimum search steps $J$ depends on the number of antennas $M$ used for the AoA search. In particular, $J\approx\frac{2M}{1.782}$ \cite{Trees2002}.}
The $i$-th vector $\bm{\gamma }_{k,i}\in\mathbb{C}^{M\times 1}, i\in \{1,\cdots,J\}$, stands for a potential AoA of user $k$ at the BS and is given by\vspace*{-0mm}
\begin{equation}
\vspace*{-1mm}
\mathbf{\gamma}_{k,i}=
\frac{1}{\sqrt{M}}\left[
\begin{array}{ccc}
1, & \ldots , & \text{ }e^{j2\pi \left( M-1\right) \tfrac{d}{\lambda }\cos
\left( \widehat{\theta }_{i}\right) }
\end{array}
\right] ^{T},
\end{equation}\vspace*{-0mm}%
where $\widehat{\theta }_{i}= \left(i-1\right)\frac{180}{J}, i\in \{1,\cdots,J\}$, is the assumed AoA and $\bm{\gamma }_{k,i}^{H}\bm{\gamma }_{k,i}=1$.
For the AoA estimation of user $k$, $\mathbf{\Gamma }_{k}$ is implemented in the $M$ phase shifters connected by the $k$-th RF chain.
The local oscillator (LO) of the $k$-th RF chain at the BS generates the same carrier frequency $f_{k}$ to down convert the received signals to the baseband, as shown in Figure \ref{fig:RF_chain}.
After the down-conversion, the signals will be filtered by a low-pass filter which can remove other frequency tones.
The equivalent received signal at the BS from user $k$ at the $i$-th potential AoA is given by\vspace*{-1mm}
\begin{equation}
\vspace*{-1mm}
r_{k,i}^{\mathrm{BS}}= \sqrt{\dfrac{\upsilon _{k}}{\upsilon
_{k}+1}}\bm{\gamma }_{k,i}^{T}\mathbf{h}_{\mathrm{L},k}^{\mathrm{BS}}+\sqrt{\dfrac{1}{\upsilon _{k}+1}}\bm{\gamma }_{k,i}^{T}\mathbf{h}_{\mathrm{S},k} +\bm{\gamma }_{k,i}^{T}\mathbf{z}_{\mathrm{BS}\vspace*{-1mm}
}.  \label{AoA_D}
\end{equation}\vspace*{-0mm}
The potential AoA, which leads to the maximum value among the $J$ observation directions, i.e.,\vspace*{-1mm}%
\begin{equation}
\vspace*{-1mm}
\widetilde{\mathbf{\gamma}}_{k}=\underset{\forall \gamma_{k,i}, \text{\ }i\in \{1,\cdots,J\}}{\arg \max \left\vert r_{k,i}^{\mathrm{BS}}\right\vert },  \label{RFBE}\vspace*{-1mm}
\end{equation}\vspace*{-0mm}%
is considered as the strongest AoA of user $k$.
{In addition, the strongest AoA estimation shown in Equations (\ref{AoA_D}) and (\ref{RFBE}) can be performed by using either a series of analog comparators in analog domain or digital buffer in digital domain.}
Besides, vector $\widetilde{\mathbf{\gamma}}_{k}$ corresponding to the AoA with the maximum value in (\ref{RFBE}) will be exploited for the design of the analog beamforming vector of user $k$ at the BS.
As a result, we can also estimate all other users' uplink AoAs at the BS from their corresponding transmitted signals simultaneously.
For notational simplicity, we denote $\mathbf{F}_{\mathrm{RF}}=\left[\begin{array}{ccc}\widetilde{\mathbf{\gamma}}_{1},\ldots, \widetilde{\mathbf{\gamma}}_{N}\end{array}\right] \in\mathbb{C}^{M\times N}$ as the BS analog beamforming matrix.

\quad\emph{Step $2$}, \underline{Line 4 in Algorithm \ref{a1}:} The BS sends unique frequency tones to all the users exploiting analog beamforming matrix\footnote{This procedure can be done simultaneously in all the RF chains for all the users.} $\mathbf{F}_{\mathrm{RF}}$ obtained in Step $1$.
This facilitates the downlink AoAs estimation at the users and this AoA information will be used to design analog beamforming vectors to be adopted at the users.
The received signal $\mathbf{y}_{k}^{\mathrm{UE}}$ at user $k$ can be expressed as\vspace*{-0.5mm}
\begin{equation}
\vspace*{-0.5mm}
\mathbf{y}_{k}^{\mathrm{UE}}=\left[ \mathbf{G}_{\mathrm{L,}k}\mathbf{h}_{\mathrm{L,}k}^{\ast }\left(\mathbf{h}_{\mathrm{L},k}^{\mathrm{BS}}\right)^{T}+\mathbf{G}_{\mathrm{S,}k}\mathbf{H}_{\mathrm{S,}k}^{T}\right]\widetilde{\mathbf{\gamma }}_{k}x_{k} +\mathbf{z}_{\mathrm{MS}},\label{RFBF}\vspace*{-0mm}
\end{equation}\vspace*{-0mm}%
where $\mathbf{z}_{\mathrm{MS}}$ denotes the thermal noise at the antenna array of the users, $\mathbf{z}_{\mathrm{MS}}\sim \mathcal{CN}\left( \mathbf{0},{\sigma_{\mathrm{MS}}^{2}}\mathbf{I}\right),$ and ${\sigma} _{\mathrm{MS}}^{2}$ is the noise variance for all the users.

The AoA detection matrix for user $k$, $\mathbf{\Omega }_{k}\in\mathbb{C}^{P\times J}$, which also contains $J$ estimation column vectors, is implemented at phase shifters of user $k$.
The $i$-th column vector of matrix $\mathbf{\Omega }_{k}$ for user $k$, $\mathbf{\omega }_{k,i}\in\mathbb{C}^{P\times 1}, i\in \{1,\cdots,J\}$, is given by\vspace*{-1mm}%
\begin{equation}
\vspace*{-1mm}
\mathbf{\omega }_{k,i}=\frac{1}{\sqrt{P}}\left[
\begin{array}{ccc}
1,  & \ldots , & e^{j2\pi \left( P-1\right) \tfrac{d}{\lambda }\cos \left(\widehat{\phi }_{i}\right) }%
\end{array}
\right] ^{T},\vspace*{-0.5mm}
\end{equation}\vspace*{-0mm}%
where $\widehat{\phi }_{i}= \left(i-1\right)\frac{180}{J}, i\in \{1,\cdots,J\}$, is the $i$-th potential AoA of user $k$ and $\mathbf{\omega }_{k,i}^{H}\mathbf{\omega }_{k,i}=1$.
With similar  procedures as shown in Step $1$, the equivalent received signal from the BS at user $k$ of the $i$-th potential AoA is given by \vspace*{-0.5mm}
\begin{align}
\vspace*{-0.5mm}
r _{k,i}^{\mathrm{UE}}= &\mathbf{\omega }_{k,i}^{H}\sqrt{\dfrac{\upsilon _{k}}{\upsilon
_{k}+1}}\mathbf{h}_{\mathrm{L,}k}^{\ast }\left(\mathbf{h}_{\mathrm{L},k}^{\mathrm{BS}}\right)^{T}\widetilde{\mathbf{\gamma }}_{k} \notag \\
+&\mathbf{\omega}_{k,i}^{H}\sqrt{\dfrac{1}{\upsilon_{k}+1}}\mathbf{H}_{\mathrm{S,}k}^{T}\widetilde{\mathbf{\gamma }}_{k}+\mathbf{\omega }_{k,i}^{H}\mathbf{z}_{\mathrm{MS}}. \label{AoA_D2}\vspace*{-0.5mm}
\end{align}\vspace*{-0.0mm}%
Similarly, we search for the maximum value among $J$ observation directions and design the analog beamforming vector based on the estimated AoA of user $k$.
The beamforming vector for user $k$ is given by\vspace*{-0.5mm}
\begin{equation}
\vspace*{-1mm}
\widetilde{\mathbf{\omega }}_{k}^{\ast}=\underset{\forall \mathbf{\omega }_{k,i}, \text{\ } i\in \{1,\cdots,J\}}{\arg \max \left\vert r
_{k,i}^{\mathrm{UE}}\right\vert }\vspace*{-1mm}
\end{equation}\vspace*{-0.0mm}and we denote $\mathbf{Q}_{\mathrm{RF}}=\left[
\begin{array}{ccc}
\widetilde{\mathbf{\omega }}_{1}^{\ast}, \ldots, \widetilde{\mathbf{\omega }}_{N}^{\ast }\end{array}\right] \in\mathbb{C}^{P\times N}$ as the users analog beamforming matrix.

\quad\emph{Step $3$}, \underline{Line 7 in Algorithm \ref{a1}:} The BS and users analog beamforming matrices based on estimated uplink AoAs and downlink AoAs are designed via Step $1$ and Step $2$, respectively.
After that, all the users transmit orthogonal pilot sequences to the BS via user beamforming vectors $\widetilde{\mathbf{\omega }}_{k}^{\ast}$.
In the meanwhile, the BS receives pilot sequences via the BS analog beamforming matrix $\mathbf{F}_{\mathrm{RF}}^{T}$.
With the analog beamforming matrices, we have the equivalent channel between the BS and the users along the strongest AoA paths\footnote{The equivalent channel consists of the BS analog beamforming matrix, the mmWave channel, and the users analog beamforming matrix.}.

We denote the pilot sequences of user $k$ in the cell as $\mathbf{\Phi }_{k}=\left[ \vartheta _{k}\left( 1\right)
,\vartheta _{k}\left( 2\right) ,....,\vartheta _{k}\left( N\right) \right]^{T}$, $\mathbf{\Phi }_{k}\in\mathbb{C}^{N\times 1}$, stands for $N$ symbols transmitted across time.
The pilot symbols used for the equivalent channel estimation are transmitted in sequence from symbol $\vartheta _{k}\left( 1\right)$ to symbol $\vartheta _{k}\left(N\right)$.
The pilot symbols for all the $N$ users form a matrix, $\mathbf{\Psi \in\mathbb{C}}^{N\times N}\mathbf{,}$ where $\mathbf{\Phi }_{k}$ is a column vector of
matrix $\mathbf{\Psi }$ given by
$\mathbf{\Psi }=\ \sqrt{E_{\mathrm{P}}}\left[
\begin{array}{ccc}
\mathbf{\Phi }_{1}, & \ldots, & \mathbf{\Phi }_{N}%
\end{array}\right]$, $\mathbf{\Phi }_{i}^{H}\mathbf{\Phi }_{j}=0$, $\forall i\neq j$, $i,\text{ }j\in \left\{ 1,\ldots, N\right\}$,
where $E_{\mathrm{P}}$ represents the transmitted pilot symbol energy.
Note that $\mathbf{\Psi }^{H}\mathbf{\Psi }=E_{\mathrm{P}}\mathbf{I}_{N}$.
Meanwhile, the BS analog beamforming matrix $\mathbf{F}_{\mathrm{RF}}$ is utilized to receive pilot sequences at all the RF chains.
As the length of the pilot sequences is equal to the number of users, we obtain an $N \times N$ observation matrix from all the RF chains at the BS.
In particular, the received signal at the $k$-th RF chain at the BS is $\mathbf{s}_{k}^{T}\in\mathbb{C}^{1\times N}$, which is given by\vspace*{-1.0mm}%
\begin{equation}
\vspace*{-1mm}
\mathbf{s}_{k}^{T}  =\widetilde{\mathbf{\gamma}}_{k}^{T}\overset{N}{\underset{i=1}%
{\sum }}\mathbf{H}_{i}\widetilde{\mathbf{\omega }}_{i}^{\ast }\sqrt{E_{%
\mathrm{P}}}\mathbf{\Phi }_{i}^{T}+\widetilde{\mathbf{\gamma}}_{k}^{T}\mathbf{Z},\label{RSIRFC}\vspace*{-1mm}
\end{equation}\vspace*{-0mm}%
where $\mathbf{Z}\in\mathbb{C}^{M\times N}$ denotes the additive white Gaussian noise matrix at the BS and the entries of $\mathbf{Z}$ are modeled by i.i.d. random variable with distribution $\mathcal{CN}\left( 0,\sigma _{\mathrm{BS}}^{2}\right)$.

After $\left[\begin{array}{ccc}\mathbf{s}_{1}, \ldots, \mathbf{s}_{N}\end{array}\right]$ is obtained, we then adopt the least square (LS) method for our equivalent channel estimation.
We note here, the LS method is widely used in practice since it does not require any prior channel information.
Subsequently, with the help of orthogonal pilot sequences, we can construct an equivalent uplink channel matrix $\widehat{\mathbf{H}}_{\mathrm{eq}}\in\mathbb{C}^{N\times N}$ formed by the proposed scheme via the LS estimation method.
Then, by exploiting the channel reciprocity, the equivalent downlink channel of the hybrid system $\widehat{\mathbf{H}}_{\mathrm{eq}}^{T}$ can be expressed as:\vspace*{-2mm}
\begin{align}
\widehat{\mathbf{H}}_{\mathrm{eq}}^{T}&=\mathbf{\Psi }^{H}\left[\begin{array}{ccc}\mathbf{s}_{1} & \ldots & \mathbf{s}_{N}\end{array}\right]
=\left[\begin{array}{c}
\widehat{\mathbf{h}}_{\mathrm{eq,}1}^{T} \\
\vdots \\
\widehat{\mathbf{h}}_{\mathrm{eq,}N}^{T}
\end{array}\right] \notag \\
&=\underset{\mathbf{H}_{\mathrm{eq}}^{T}}{\underbrace{\left[
\begin{array}{c}
\widetilde{\mathbf{\omega }}_{1}^{H}\mathbf{H}_{1}^{T}\mathbf{F}_{\mathrm{RF}%
} \\
\vdots \\
\widetilde{\mathbf{\omega}}_{N}^{H}\mathbf{H}_{N}^{T}\mathbf{F}_{\mathrm{RF}}\end{array}\right]}}+\underset{\mathrm{effictive}\text{\ }\mathrm{noise}}{\underbrace{\frac{1}{\sqrt{E_{\mathrm{P}}}}\left[
\begin{array}{c}
\mathbf{\Phi}_{1}^{H}\mathbf{Z}^{T}\mathbf{F}_{\mathrm{RF}} \\
\vdots \\
\mathbf{\Phi} _{N}^{H}\mathbf{Z}^{T}\mathbf{F}_{\mathrm{RF}}
\end{array}\right] }},  \label{EHC_1}\vspace*{-2mm}
\end{align}\vspace*{-0mm}%
where $\widehat{\mathbf{h}}_{\mathrm{eq,}k}$ is the $k$-th column vector of matrix $\widehat{\mathbf{H}}_{\mathrm{eq}}$.
From Equation (\ref{EHC_1}), we observe that the proposed channel estimation algorithm can obtain all users' equivalent CSI simultaneously.

\vspace*{-2mm}
\subsection{Performance Analysis of Proposed Channel Estimation}
In the high SNR regime, the effective noise component is negligible and Equation (\ref{EHC_1}) can be simplified\footnote{The performance degradation due to the high SNR assumption will be verified by analysis and simulation in the following sections.} as\vspace*{-1mm}
\begin{align}
\vspace*{-2mm}
&\mathbf{H}_{\mathrm{eq}}^{T}=\left[
\begin{array}{c}
\widetilde{\mathbf{\omega }}_{1}^{H}\mathbf{H}_{1}^{T}\mathbf{F}_{\mathrm{RF}%
} \\
\vdots \\
\widetilde{\mathbf{\omega }}_{N}^{H}\mathbf{H}_{N}^{T}\mathbf{F}_{\mathrm{RF}%
}%
\end{array}%
\right] \notag \\
&=\left[
\begin{array}{ccc}
\widetilde{\mathbf{\omega }}_{1}^{H} & \cdots & \mathbf{0} \\
\vdots & \ddots  & \vdots \\
\mathbf{0} & \cdots & \widetilde{\mathbf{\omega }}_{N}^{H}%
\end{array}\right]
\left[
\begin{array}{c}
\mathbf{H}_{1}^{T} \\
\vdots \\
\mathbf{H}_{N}^{T}
\end{array}
\right]
\underset{\mathbf{F}_{\mathrm{RF}}}{\underbrace{\left[
\begin{array}{ccc}
\widetilde{\mathbf{\gamma}}_{1} & \ldots & \widetilde{\mathbf{\gamma}}_{N}\end{array}\right] }}.  \label{EHC_2}\vspace*{-2mm}
\end{align}\vspace*{-0mm}%
From Equation (\ref{EHC_2}), we can see that without the impact of noise, we can perfectly estimate the equivalent channels, which consists of the BS analog beamforming matrix $\mathbf{F}_{\mathrm{RF}}$, the users analog beamforming matrix $\mathbf{Q}_{\mathrm{RF}}^{T}$, and mmWave channels between the BS and all the users.

Now, we analyze the performance of the proposed channel estimation.
With the normalization factor $\frac{1}{\sqrt{MP}}$, the normalized mean square error (MSE) of equivalent channel estimation $\widehat{\mathbf{h}}_{\mathrm{eq,}k}^{T}$ is given by\vspace*{-0mm}
\begin{align}
\vspace*{-0mm}
\mathrm{MSE}_{\mathrm{eq,}k}& =\frac{1}{N}\mathrm{tr}\left\{
\mathrm{E}_{\mathbf{h}_{\mathrm{S},k}}\left[ \left( \dfrac{1}{\sqrt{MP}}\widehat{\mathbf{h}}_{\mathrm{eq,%
}k}^{T}-\dfrac{1}{\sqrt{MP}}\mathbf{h}_{\mathrm{eq},k}^{T}\right) ^{H} \right. \right. \notag \\
&\left. \left. \left(\dfrac{1}{\sqrt{MP}}\widehat{\mathbf{h}}_{\mathrm{eq,}k}^{T}-\dfrac{1}{\sqrt{MP}}\mathbf{h}_{\mathrm{eq},k}^{T}\right) \right] \right\}  \notag \\
& =\frac{\sigma _{\mathrm{BS}}^{2}\mathrm{tr}\left[ \mathbf{F}_{\mathrm{RF}%
}^{H}\mathbf{F}_{\mathrm{RF}}\right] }{E_{\mathrm{P}}NMP}=%
\frac{\sigma _{\mathrm{BS}}^{2}}{E_{\mathrm{P}}MP}.  \label{MSE_NN}\vspace*{-0mm}
\end{align}\vspace*{-0mm}%
From (\ref{MSE_NN}), we observe that the normalized MSE of the equivalent channel of user $k$ decreases with an increasing transmitted pilot symbol power $E_{\mathrm{P}}$ as well as the numbers of antennas equipped at the BS and at each user.
As the numbers of antennas $M\ $and $P$ approach infinity, the impact of noise will vanish asymptotically.
In contrast, the channel estimation errors caused by noise in conventional fully digital massive MIMO systems cannot be mitigated by increasing the number of antennas equipped at the BS and the users \cite{Marzetta2010,Jose2011}.
Therefore, the proposed channel estimation for the hybrid system outperforms the conventional pilot-aided channel estimation for the fully digital system in terms of noise mitigation.
The MSE analysis result will be verified via simulation.
Here, we would point out that the proposed channel estimation scheme can exploit large array gains offered by the antenna arrays via using the BS analog beamforming matrix $\mathbf{F}_{\mathrm{RF}}$ and the users analog beamforming matrix $\mathbf{Q}_{\mathrm{RF}}$ to enhance the receive SNR of pilot symbols \cite{book:wireless_comm}.
Therefore, it is expected that the performance of the proposed channel estimation scheme improves with increasing the numbers of antennas equipped at the BS and each user.

For a multi-cell scenario, the received pilot sequences at the RF chains of the desired BS will be affected by the reused pilot sequences from neighboring cells \cite{Marzetta2010,Jose2011,WU2016}. The threat of pilot contamination attack will be detailed and discussed in our future work.
\begin{figure}[t]
\centering\vspace{-1mm}
\includegraphics[width = 3.5 in]{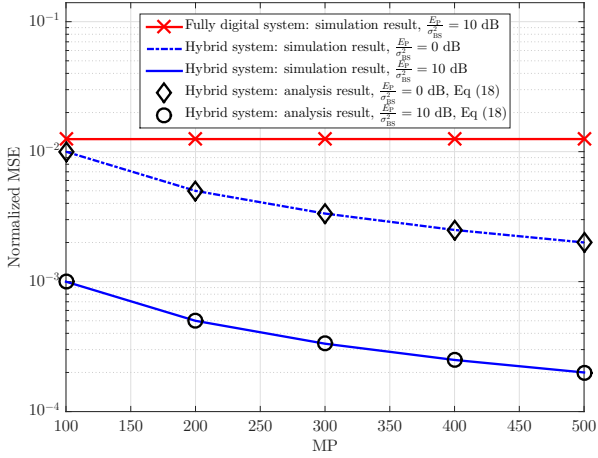}\vspace{-0mm}
\caption{The normalized MSE performance comparison between the proposed
pilot-aided channel estimation algorithm for the hybrid system and the conventional
pilot-aided LS for the fully digital system.}
\label{fig:antiPC}\vspace{-1mm}
\end{figure}

Figure \ref{fig:antiPC} shows the normalized MSE of the equivalent channel estimation of user $k$ versus the total numbers of antennas equipped at the BS and the user.
In Figure \ref{fig:antiPC}, we see that the simulation results match with analytical results derived in Equation (\ref{MSE_NN}).
For the SNR of the transmitted pilot at $\frac{E_{\mathrm{P}}}{\sigma _{\mathrm{BS}}^{2}}=10$ dB, it can be seen that the normalized MSE of the proposed algorithm decreases with an increasing $MP$.
On the contrary, we have the observation that the numbers of antennas equipped at the BS and users have no significant impact on the conventional pilot-aided LS channel estimation for the fully digital system.
Furthermore, Figure \ref{fig:antiPC} also shows that a high transmitted pilot energy is helpful to lower the normalized MSE of the proposed algorithm, since the impact of noise is reduced proportional to $\frac{1}{E_{\mathrm{P}}MP}$ as shown in (\ref{MSE_NN}).
It is interesting to note that, to meet a certain required MSE of channel estimation, we can either increase the number of antennas equipped at the BS or the number of antennas at each user. This indicates that increasing the number of antennas equipped at the BS can always improve the system performance, despite the possibly limited numbers of antennas equipped at the users.

\vspace*{-2mm}
\section{ZF Precoding and Performance Analysis}
In this section, we illustrate and analyze the achievable rate performance per user of the considered hybrid mmWave system under digital ZF downlink transmission.
The digital ZF downlink precoding is based on the estimated equivalent channel $\widehat{\mathbf{H}}_{\mathrm{eq}}$, which subsumes the BS analog beamforming matrix $\mathbf{F}_{\mathrm{RF}}$ and the users analog beamforming matrix $\mathbf{Q}_{\mathrm{RF}}$.
We derive a closed-form upper bound of achievable rate per user of the ZF precoding in the considered hybrid system.
Also, we compare the system achievable rate upper bound obtained by the fully digital system exploiting the ZF precoding for a large number of antennas.
The achievable rate performance gap between the considered hybrid mmWave system and the fully digital system is characterized, which is verified via analysis and simulation results.
\vspace*{-5mm}
\subsection{ZF Precoding}

Now, we utilize the estimated equivalent channel for downlink digital ZF precoding.
To study the best achievable rate performance, we first assume that the equivalent channel is estimated in the high SNR regime.
In this case, the equivalent channel is considered as perfectly estimated as the proposed channel estimation is only affected by noise as shown in Equation (\ref{MSE_NN}).
Therefore, the baseband digital ZF precoder $\overline{\mathbf{W}}_{\mathrm{eq}}\in\mathbb{C}^{N\times N}$ based on $\mathbf{H}_{\mathrm{eq}}$ is given by
\begin{equation}
\vspace*{-0mm}
\overline{\mathbf{W}}_{\mathrm{eq}}=\mathbf{H}_{\mathrm{eq}}^{\ast }(\mathbf{H}_{\mathrm{eq}}^{T}\mathbf{H}_{\mathrm{eq}}^{\ast})^{-1}=\left[\begin{array}{ccc}\overline{\mathbf{w}}_{\mathrm{eq,}1},\ldots ,\overline{\mathbf{w}}_{\mathrm{eq,}N}
\end{array}%
\right] ,  \label{P1}\vspace*{-0.0mm}
\end{equation}%
where $\overline{\mathbf{w}}_{\mathrm{eq,}k}\in\mathbb{C}^{N \times 1}$ is the $k$-th column of ZF precoder for user $k$.
As each user is equipped with only one RF chain, one superimposed signal is received at each user at each time instant with hybrid transceivers.
The received signal at user $k$ after beamforming can be expressed as:
\vspace*{-0mm}
\begin{align}
\vspace*{-1mm}
y_{\mathrm{ZF}}^{k}=&\underset{\mathrm{desired}\text{ }\mathrm{signal}}{%
\underbrace{\widetilde{\mathbf{\omega }}_{k}^{H}\mathbf{H}_{k}^{T}\mathbf{F}%
_{\mathrm{RF}}\overline{\beta }\overline{\mathbf{w}}_{\mathrm{eq,}k}x_{k}}} \notag \\
&+\underset{\mathrm{interference}}{\underbrace{\widetilde{\mathbf{\omega }}_{k}^{H}\mathbf{H}_{k}^{T}\overset{N}{\underset{j=1,j\neq k}{\sum }}\mathbf{F}%
_{\mathrm{RF}}\overline{\beta }\overline{\mathbf{w}}_{\mathrm{eq,}j}x_{j}}}+%
\underset{\mathrm{noise}}{\underbrace{\widetilde{\mathbf{\omega }}_{k}^{H}%
\mathbf{z}_{\mathrm{MS},k}}},  \label{P2}\vspace*{-1mm}
\end{align}\vspace*{-0mm}%
where $x_{k}\in\mathbb{C}^{1\times 1}$ is the transmitted symbol energy from the BS to user $k$, $\mathrm{E}\left[ \left\vert x_{k}^{2}\right\vert \right]=E_{s}$, $E_{s}$ is the average transmitted power for each user, $\overline{\beta }=\sqrt{\tfrac{1}{\mathrm{tr}(\overline{\mathbf{W}}%
_{\mathrm{eq}}\overline{\mathbf{W}}_{\mathrm{eq}}^{H})}}$ is the
transmission power normalization factor, and the effective
noise part $\mathbf{z}_{\mathrm{MS,}%
k}\sim \mathcal{CN}\left( \mathbf{0},{\sigma_{\mathrm{MS}}^{2}}\mathbf{I}\right) $.
Due to the fact that the MU interference within the AoA directions can be suppressed by the digital ZF precoder, thus\vspace*{-0mm}%
\begin{align}
\vspace*{-2mm}
&\mathbf{h}_{\mathrm{eq},i}^{T}\overline{\mathbf{w}}_{\mathrm{eq,}j}=0,\text{
}\forall i\neq j, \text{ } \mbox{and}\notag \\
&\widetilde{\mathbf{\omega }}_{k}^{H}\mathbf{H}_{k}^{T}\overset{N}{\underset{j=1,j\neq k}{\sum }}\mathbf{F}%
_{\mathrm{RF}}\left( \overline{\mathbf{w}}_{\mathrm{eq,}j}\right) x_{j}=0.
\label{P3}\vspace*{-2mm}
\end{align}\vspace*{-0mm}%
Then we express the signal-to-interference-plus-noise ratio (SINR) of user $k$ as \vspace{-1mm}
\begin{equation}
\mathrm{SINR}_{\mathrm{ZF}}^{k}=\frac{\overline{\beta }^{2}E_{s}}{\sigma _{%
\mathrm{MS}}^{2}}.  \label{Eq_1520}\vspace*{-1mm}
\end{equation}\vspace{-0mm}%
In the sequal, we study the performance of the considered hybrid mmWave systems.
For simplicity, we assume that channels of all the users have the same Rician K-factor, i.e., $\upsilon _{k}=\upsilon,\forall k$.

\vspace*{-2mm}
\subsection{Performance Upper Bound of ZF Precoding}
Now, exploiting the SINR expression in (\ref{Eq_1520}), we summarize the upper bound of achievable rate per user of the digital ZF precoding with the proposed channel estimation algorithm in the following theorem.\vspace*{-0.0mm}
\begin{theo}\label{thm:Theo_1}
The achievable rate per user of the proposed ZF precoding is bounded by \vspace*{-0mm}%
\begin{align}
R_{\mathrm{HB}} \leqslant R_{\mathrm{HB}}^{\mathrm{upper}}= & \log _{2}\left\{ 1+\left[ \left( \dfrac{\upsilon}{\upsilon +1}\right)MP  \| \mathbf{F}_{\mathrm{RF}}^{H}\mathbf{F}_{\mathrm{RF}} \|_{\mathrm{F}}^{2} \right. \right. \notag \\
&\left. \left. +\left( \dfrac{1}{\upsilon+1}\right) N^{2}\right] \frac{1}{N^{2}} \dfrac{E_{s}}{\sigma _{\mathrm{MS}}^{2}}\right\}.  \label{Theo_1}
\end{align}
\end{theo}
\begin{proof}
Please refer to Appendix A.
\end{proof}
\vspace*{+1mm}
From Equation (\ref{Theo_1}), we see that the upper bound of achievable rate per user of the proposed ZF precoding depends on the Rician K-factor, $ \upsilon$.
Also, we can further observe that the upper bound of the achievable rate per user also depends on the BS analog beamforming matrix $\mathbf{F}_{\mathrm{RF}}$ designed in Step $2$ of the proposed channel estimation algorithm.
We note that since the analog beamforming only allows the BS to transmit each user's signal via its strongest AoA direction, the proposed scheme can utilize the transmission power more effectively.
In addition, the interference outside the strongest AoA directions is reduced as less transmission power will leak to undesired users.
On other hand, with an increasing number of antennas at the BS, the communication channels are more likely to be orthogonal to each other.
Therefore, it is interesting to evaluate the asymptotic upper bound $R_{\mathrm{HB}}^{\mathrm{upper}}$ for the case of a large number of antennas. We note that, even if the number of antennas equipped at the BS is sufficiently large, the required number of RF chains is still only required to equal to the number of users in the hybrid mmWave system structures. 
\vspace*{-0mm}
\begin{coro}
In the large numbers of antennas regime, i.e., $M\rightarrow\infty $, such that $\mathbf{F}_{\mathrm{RF}}^{H}\mathbf{F}_{\mathrm{RF}}\overset{a.s.}{\rightarrow }$ $\mathbf{I}_{N},$ the asymptotic achievable rate per user of the hybrid system is bounded above by
\begin{equation}
\hspace*{-2mm}
R_{\mathrm{HB}}^{\mathrm{upper}}\underset{M\rightarrow \infty }{\overset{a.s.}{\rightarrow }}\log _{2}\left\{ 1+\left[
\frac{MP}{N}\frac{\upsilon }{\upsilon +1} +\dfrac{1}{\upsilon +1}\right] \dfrac{E_{s}}{\sigma _{\mathrm{MS}}^{2}}\right\}.
\label{HSUB_LA}\vspace*{-0mm}
\end{equation}\label{Coro_1}
\end{coro}\vspace{-2mm}
\begin{proof}
The result follows by substituting $\mathbf{F}_{\mathrm{RF}}^{H}\mathbf{F}_{\mathrm{RF}}\underset{M\rightarrow \infty }{\overset{a.s.}{\rightarrow }}\mathbf{I}_{N}$ into (\ref{Theo_1}).
\end{proof}
\vspace*{+1mm}
From Equation (\ref{HSUB_LA}), we have an intuitive observation that the asymptotic performance of the proposed precoding is mainly determined by the numbers of equipped antennas and RF chains.
\vspace*{-2mm}
\subsection{Comparison with Fully Digital Systems}

In this section, we derive the achievable rate performance of a fully digital mmWave system in the large numbers of antennas regime.
The obtained analytical results in this section will be used as a reference for comparing to the proposed hybrid system.
To this end, for the fully digital mmWave system, we assume that each user is equipped with one RF chain and $P$ antennas.
The $P$ antenna array equipped at each user can provide $10\log(P)$ dB array gain.
We note that, the number of antennas equipped at the BS is $M$ and the number of RF chains equipped at the BS is equal to the number of antennas.
The channel matrix from the BS to user $k$ is given by\vspace*{-0.0mm}
\begin{equation}
\vspace*{-0mm}
\mathbf{H}_{k}^{T}=\mathbf{h}_{k}^{\ast }\mathbf{h}_{\mathrm{BS,}k}^{T}.\vspace*{-0mm}
\end{equation}\vspace*{-0mm}%
We assume that the CSI is perfectly known to the users and the BS.
The BS with the fully digital system is adopted to illustrate the maximal performance gap in terms of achievable rate between the fully digital system and the considered hybrid system.
The CSI at the BS for the downlink information transmission is given by\vspace*{-0.0mm}
\begin{equation}
\vspace*{-0mm}
\mathbf{H}_{\mathrm{FD}}^{T}=\left[
\begin{array}{ccccc}
{\mathbf{h}}_{\mathrm{BS,}1}^{T} & ... & {\mathbf{h}}_{%
\mathrm{BS,}k}^{T} & ... & {\mathbf{h}}_{\mathrm{BS,}N}^{T}%
\end{array}%
\right] .\vspace*{-0mm}
\end{equation}\vspace*{-0.0mm}%
The ZF precoder for the equivalent channel $\mathbf{H}_{\mathrm{FD}}^{T}$ is
denoted as\vspace*{-1.5mm}
\begin{equation}
\mathbf{W}_{\mathrm{FD}}=\mathbf{H}_{\mathrm{FD}}^{\ast }\left( \mathbf{H}_{%
\mathrm{FD}}^{T}\mathbf{H}_{\mathrm{FD}}^{\ast }\right) ^{-1}.
\end{equation}\vspace*{-0mm}%
Therefore, the achievable rate per user of the fully digital system is bounded by
\begin{align}
R_{\mathrm{FD}}&=\log _{2}\left[ 1+\dfrac{P}{\mathrm{tr}\left[ \mathbf{W}_{\mathrm{FD}}\mathbf{W}_{\mathrm{FD}}^{H}\right] }\dfrac{E_{s}}{\sigma _{\mathrm{MS}}^{2}}\right] \notag \\
& \overset{(c)}{\leqslant }\log _{2}\left[ 1+\frac{P }{N^{2}}\mathrm{tr}\left[ \mathbf{H}_{\mathrm{FD}}^{H}\mathbf{H}_{\mathrm{FD}}\right] \dfrac{E_{s}}{\sigma_{\mathrm{MS}} ^{2}}\right] =R_{\mathrm{FD}}^{\mathrm{upper}},  \label{FDUB}
\end{align}
where $(c)$ follows (\ref{Proof_3}) in Appendix A.%
\vspace*{-0mm}
\begin{coro}
In the large numbers of antennas regime, the asymptotic achievable rate per user of the fully digital system is bounded by\vspace*{-1mm}
\begin{equation}
\vspace*{-1mm}
R_{\mathrm{FD}}\leqslant R_{\mathrm{FD}}^{\mathrm{upper}}\underset{%
M\rightarrow \infty }{\overset{a.s.}{\rightarrow }}\log _{2}\left[ 1+\frac{MP%
}{N}\frac{E_{s}}{\sigma _{\mathrm{MS}}^{2}}\right] .  \label{FDUB_LA}\vspace*{-0mm}
\end{equation}
\end{coro}\vspace*{-0mm}
\begin{proof}
The result follows by substituting $\frac{1}{M}\mathbf{H}_{\mathrm{FD}}^{H}%
\mathbf{H}_{\mathrm{FD}}\underset{M\rightarrow \infty }{\overset{a.s.}{%
\rightarrow }}\mathbf{I}_{N}$ into (\ref{FDUB}).
\end{proof}
\vspace{+1mm}
\begin{figure}[t]
\centering\vspace*{-1mm}
\begin{subfigure}{.5\textwidth}
  \centering
  \includegraphics[width=3.5in,]{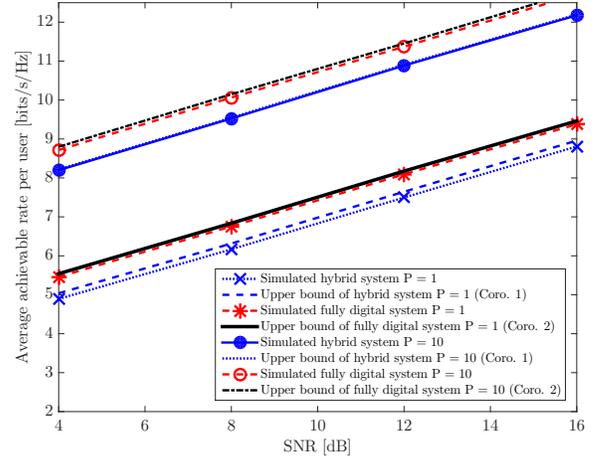}
  \caption{Average achievable rate versus SNR.}
  \label{fig:HBvsFD}
\end{subfigure}%
\vfill
\begin{subfigure}{.5\textwidth}
  \centering
  \includegraphics[width=3.5in,]{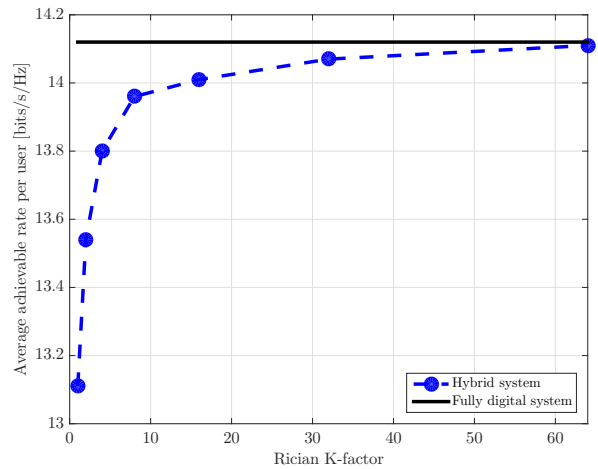}
  \caption{Average achievable rate versus Rician K-factor.}
  \label{fig:HBFDvsRK}
  \end{subfigure}
\vspace*{-0mm}
\caption{(a) Average achievable rate per user (bits/s/Hz) versus SNR for the hybrid and the fully digital systems with $N=10$ and $\upsilon =2$. (b) Average achievable
rate (bits/s/Hz) versus Rician K-factor for different systems with $\mathrm{SNR}=20$ dB, $P=16$, and $N=4$. We set the number of antennas $M=100$ for both (a) and (b).}
\label{fig:HBFD}\vspace{-1mm}
\end{figure}

Based on (\ref{HSUB_LA}) and (\ref{FDUB_LA}), we further quantify the achievable rate performance gap between the considered hybrid system and the fully digital system in the large numbers of antennas regime.\vspace{-0mm}%
\begin{coro}
In the large numbers of antennas regime, the gap between the achievable rate upper bounds for the hybrid system and the fully digital system can be expressed as\vspace{-1mm}:
\begin{equation}
\vspace*{-1mm}
\Delta_{\mathrm{GAP}}=R_{\mathrm{HB}}^{\mathrm{upper}}-R_{\mathrm{FD}}^{\mathrm{upper}}\underset{M\rightarrow \infty }{\overset{\mathrm{(S)}}{\approx }}\log _{2}\left\{ \frac{\upsilon}{\upsilon+1}\right\} \leqslant 0,\label{GAP_LA}\vspace*{-1mm}
\end{equation}\vspace{-0mm}%
where $\mathrm{(S)}$ {}stands for $\mathrm{SNR\rightarrow \infty .}$ \label{coro3}
\end{coro}
\vspace*{-0mm}
\begin{proof}
The result follows by substituting (\ref{HSUB_LA}) and (\ref{FDUB_LA}) into (\ref{GAP_LA}).
\end{proof}
\vspace*{+1mm}
In the large numbers of antennas regime, based on (\ref{HSUB_LA}) and (\ref{FDUB_LA}), it is interesting to observe that with an increasing Rician K-factor $\upsilon $, the performance upper bounds of the two considered systems will coincide.
Intuitively, as the Rician K-factor increases, the LOS component becomes the dominant element of the communication channel, as shown in Equation \eqref{eqn:LOS_channel}.
Therefore, the BS analog beamforming matrix based on the estimated strongest AoA will allocate a smaller portion of the transmitted signal energy to the scattering component.
At the same time, the interference caused by other users is suppressed by the digital baseband ZF precoding, $\overline{\mathbf{W}}_{\mathrm{eq}}$.

In Figure \ref{fig:HBvsFD}, we present a comparison between the achievable rate per user of the hybrid system and the fully digital system for $M=100,$ $N=10$, and a Rician K-factor of $\upsilon_{k} =2,\forall k$.
Firstly, our simulation results verify the tightness of derived upper bounds in (\ref{HSUB_LA}) and (\ref{FDUB_LA}).
It can be observed from Figure \ref{fig:HBvsFD} that even for a small value of Rician K-factor, e.g. $\upsilon =2$, our proposed channel estimation scheme with ZF precoding can achieve a considerable high sum rate performance due to its interference suppression capability.
In Figure \ref{fig:HBFDvsRK}, the achievable rate performance gap between the fully digital system and the hybrid system decreases with the increasing Rician K-factor, which is predicted by Equation (\ref{GAP_LA}).
In particular, with a sufficiently large Rician K-factor, the achievable rates of these two systems will coincide.

\vspace*{-0mm}
\section{Performance Analysis with Hardware Impairments}

In the last section, we study the proposed mmWave hybrid system in ideal hardware and ideal estimation conditions.
{In practice, hardware components may have various types of impairments that may degrade the achievable rate performance\footnote{{The authors of \cite{Bjornson2015bb} proved that for a fully digital massive MIMO system, the additive distortion caused by hardware impairments create finite ceilings on the channel
estimation accuracy and on the uplink/downlink capacity, which are irrespective of the SNR and the
number of base station antennas.
In addition, work \cite{Ying2015} concluded that the impact of phase error on hybrid beamforming is a further reduction on the number of effective antenna per user.
In fact, the achievable rate degradation caused by phase errors can be compensated by simply employing more transmit antennas.}}, e.g. phase errors in phase shifters induced by thermal noise, transceiver RF beamforming errors caused by AoA estimation errors, and channel estimation errors affected by independent additive distortion noises \cite{Bjornson2015bb,Zhang2016,Ying2015}.}
In this section, we analyze the rate performance degradation under hardware impairments.
\vspace*{-5mm}
\subsection{Transceiver Beamforming Errors and Random Phase Errors}

Here, we first discuss the scenario that the equivalent channels are estimated in the high pilot transmit power regime, i.e., $\frac{E_{\mathrm{P}}}{\sigma _{\mathrm{BS}}^{2}}\rightarrow \infty $, where errors caused by thermal noise are negligible.
As a result, we focus on the effect of the phase errors in phase shifters and transceiver analog beamforming errors on the achievable rate performance.

Firstly, we start from quantifying the impact of transceiver beamforming errors, which is caused by AoA estimation errors.
As the number of antennas equipped at the BS is sufficiently large, the beamwidth of the antenna array is narrow.
Therefore, even a small AoA estimation error may cause significant impacts on the system performance.
Here, the BS analog beamforming error matrix of user $k$, $\mathbf{\Delta  }_{k}\in\mathbb{C}^{M\times M}$, which is caused by the AoA estimation error, can be expressed as %
\begin{equation}
\vspace*{-1.5mm}
\mathbf{\Delta }_{k}=\frac{1}{\sqrt{M}}\left[
\begin{array}{ccc}
1 &  &  \\
& \ddots &  \\
&  & e^{j2\pi \left( M-1\right) \tfrac{d}{\lambda }\cos \left( \Delta \theta
_{\mathrm{BS,}k}\right) }%
\end{array}%
\right] ,
\end{equation}\vspace*{-0mm}%
where the AoA estimation error $\Delta \theta _{\mathrm{BS,}k}$ at the BS for user $k$ is modeled by i.i.d. random variable with the distribution $\mathcal{CN}\left( 0, \varrho _{\mathrm{BS}}^{2}\right)$.
In addition, the user beamforming error matrix $\mathbf{\Theta }_{k}\in\mathbb{C}^{P\times P}$, which is caused by AoA estimation error $\Delta \theta _{\mathrm{MS,}k}$, can also be similarly formulated.
Now, we assume that the hybrid system is only affected by AoA estimation errors.
Therefore, the received pilot symbols at the the BS during the third step of channel estimation is expressed as:\vspace*{-0.0mm}
\begin{equation}
\vspace*{-1mm}
\left(\widehat{\mathbf{s}}_{k}^{\mathrm{BE}}\right)^{T}=\widetilde{\mathbf{\gamma}}_{k}^{T}\mathbf{\Delta }%
_{k}\overset{N}{\underset{i=1}{\dsum }}\mathbf{H%
}_{i}\mathbf{\Theta }_{i}\widetilde{\mathbf{\omega }}_{%
i}^{\ast }\sqrt{E_{\mathrm{P}}}\mathbf{\Phi }_{i}^{T}+\underset{\mathrm{effective}\text{ }\mathrm{noise}}{\underbrace{\widehat{\mathbf{z}}_{\mathrm{eq}}}},\vspace*{-0mm}
\end{equation}\vspace*{-0mm}%
where the entries of effective noise part $\widehat{\mathbf{z}}_{\mathrm{eq}}=\mathbf{f}_{k}^{T}\mathbf{\Delta }_{k}\mathbf{\Xi }_{k}\mathbf{Z}$ can still be modeled by i.i.d. random variable with distribution $\mathcal{CN}\left( 0,\sigma _{\mathrm{BS}}^{2}\right)$.
As we assumed that $\frac{E_{\mathrm{P}}}{\sigma _{\mathrm{BS}}^{2}}\rightarrow \infty $, the estimated equivalent channel under transceiver analog beamforming errors, $\mathbf{H}_{\mathrm{eq}}^{\mathrm{BE}}$, is given by%
\begin{align}
\vspace*{-1mm}
&\left( \mathbf{H}_{\mathrm{eq}}^{\mathrm{BE}}\right) ^{T}= \\
&\left[
\begin{array}{ccc}
\widetilde{\mathbf{\omega }}_{1}^{H}\mathbf{\Theta }_{1}&\cdots&\mathbf{0%
} \\
\vdots & \ddots & \vdots \\
\mathbf{0}&\cdots&\widetilde{\mathbf{\omega }}_{N}^{H}\mathbf{\Theta }%
_{N}%
\end{array}%
\right] \left[
\begin{array}{c}
\mathbf{H}_{1}^{T} \\
\vdots \\
\mathbf{H}_{N}^{T}%
\end{array}%
\right] \left[
\begin{array}{ccc}
\mathbf{\Delta }_{1}\widetilde{\mathbf{\gamma}}_{1}\ldots\mathbf{\Delta }%
_{N}\widetilde{\mathbf{\gamma}}_{N}%
\end{array}%
\right].\notag
\end{align}\vspace*{-0mm}%
Following the similar signal processing procedures as in (\ref{P1})--(\ref{P3}), the average achievable rate per user under the BS and users analog beamforming errors can be expressed as%
\begin{equation}
\vspace*{-0.5mm}
\widehat{R}_{\mathrm{HB}}^{\mathrm{BE}}=\mathrm{E}_{\mathbf{H}_{\mathrm{S}},\mathrm{\Delta \theta }_{\mathrm{MS}}\mathrm{,\Delta \theta }_{%
\mathrm{BS}}}\left\{ \log _{2}\left[ 1+\frac{(\widehat{\beta}^{\mathrm{BE}})^{2}E_{s}}{%
\sigma _{\mathrm{MS}}^{2}}\right] \right\} , \label{HBBFE}\vspace*{-0mm}
\end{equation}\vspace*{-0mm}%
where $\widehat{\beta}^{\mathrm{BE}}=\sqrt{\frac{1}{\mathrm{tr}[(\widehat{\mathbf{W}}_{\mathrm{eq}}^{\mathrm{BE}})(\widehat{\mathbf{W}}_{\mathrm{eq}}^{\mathrm{BE}})^{H}]}}$ is the transmission power normalization factor and $\widehat{\mathbf{W}}_{\mathrm{eq}}^{\mathrm{BE}}$ is the downlink ZF precoder based on $\mathbf{H}_{\mathrm{eq}}^{\mathrm{BE}}$.
From (\ref{HBBFE}), it is interesting to observe that the achievable rate performance is not bounded by above.
This is due to the fact that the impact of transceiver analog beamforming errors on different RF chains can be estimated by the pilot matrix $\mathbf{\Psi}$, treated as parts of the channel, and compensated by the digital ZF transmission.
Therefore, it is expected that our proposed scheme is robust against transceiver analog beamforming errors.

Observing Equations (\ref{HBBFE}) and (\ref{Eq_1520}), we see that the performance degradation caused by AoA estimation errors can be quantified by comparing power normalization factors $\widehat{\beta}^{\mathrm{BE}}$ to $\overline{\beta}$.
In particular, the hybrid system with hardware impairments incurs a power loss in the received SINR at the users.
However, it is difficult to provide an explicit mathematical expression to quantify the loss.
Therefore, we borrow the ``power loss'' concept from the literature of array signal processing.
Specifically, the power loss due to AoA estimation errors is related to the half power beamwidth (HPBW) \cite{Trees2002} and the HPBW can be approximated [p. 48, 33] %
\begin{equation}
\vspace*{-1.0mm}
\mathrm{HPBW} \approx \dfrac{1.782}{M},\vspace*{-0mm}
\end{equation}\vspace*{-0mm}%
where $M$ is the number of antennas equipped in the array.
Hence, we can approximate the impact of AoA estimation errors by introducing a power loss coefficient $\xi \in \left( 0,1\right]$, which is determined by the beam pattern of the array and the variance of the distribution of AoA estimation errors.
For example, if the variance of AoA estimation errors is assumed no larger than half of the HPBW of the antenna array\footnote{In general, this assumption is valid for practical communication systems with a sufficiently large number of antennas.}, i.e., $\varrho _{\mathrm{BS}}^{2}\leqslant \frac{1.782}{2M}$, the power loss coefficient is given by $\xi \approx 0.5$ according to the half power loss principle \cite{Trees2002}. 
With an increasing AoA estimation error variance $\varrho _{\mathrm{BS}}^{2}$, $\xi $ decreases significantly\footnote{The calculation of $\xi $ for different numbers of antennas and AoA errors is well studied for linear arrays and interested readers may refer to \cite{Trees2002} for a detailed discussion.}. 
With the help of $\xi$, we now can express the approximation of $\widehat{R}_{\mathrm{HB}}^{\mathrm{BE}}$, the average achievable rate per user under transceiver beamforming errors, in the large numbers of antennas regime as
\vspace*{-0mm}%
\begin{equation}
\vspace*{-1mm}
\hspace*{-0mm}
\widehat{R}_{\mathrm{HB}}^{\mathrm{BE}}\overset{M\rightarrow \infty }{%
{\approx }}\log _{2}\left\{ 1+\left[
\dfrac{\upsilon }{\upsilon +1}\frac{MP}{N}\xi +\dfrac{1}{\upsilon +1}\right] \dfrac{E_{s}}{\sigma _{\mathrm{MS}}^{2}}\right\}.\label{HPBW_D}\vspace*{+1mm}
\end{equation}\vspace*{-0mm}%
From (\ref{HPBW_D}), it is interesting to note that, the beamforming errors matrices $\mathbf{\Theta }_{k}$ and $\mathbf{\Delta }_{k}$ only lead to a certain power loss, which depends on $\xi$. Besides, this loss can be compensated at the expense of a higher transmit power, i.e., increase $E_{\mathrm{S}}$.

Now, we discuss the impact of phase errors, which are caused by the additive white Gaussian noise (AWGN) and the limited quantization resolution of the phase shifters.
At the users sides, phase errors of user $k$ are modeled by \cite{Ying2015}\vspace*{-0mm}
\begin{equation}
\vspace*{-0mm}
\mathbf{\Lambda }_{k}=\mathrm{diag}\left\{ e^{j\Delta \phi
_{k,p}}\right\}\in \mathbb{C}^{P\times P} ,\text{ }p\in \left\{ 1,\ldots ,P\right\},\vspace*{-0mm}
\end{equation}\vspace*{-0mm}%
where phase errors $\Delta \phi _{k,p}, \forall p$, are uniformly distributed over $\left[ -a,\text{ }a\right] \ $and $a>0$ is the maximal phase error of user $k$.
Similarly, for the phase shifters connected with the $k$-th RF chain in the BS, the associated phase errors are given by\vspace*{-0mm}
\begin{equation}
\vspace*{-0mm}
\mathbf{\Xi }_{k}=\mathrm{diag}\left\{ e^{j\Delta \psi
_{k,m}}\right\}\in\mathbb{C}^{M\times M} ,\text{ }m\in \left\{ 1,\ldots ,M\right\} ,\vspace*{-0mm}
\end{equation}\vspace*{-0mm}%
where the errors $\Delta\psi _{k,m}, \forall m$, are uniformly distributed over $\left[-b,\ b\right]$ and $b>0$ is the maximal phase error of the BS.
The property of $\mathbf{\Lambda }_{k}$ and $\mathbf{\Xi }_{k}$ can be expressed as\vspace*{-0.5mm}%
\begin{align}
\vspace*{-1mm}
&\mathrm{E}\left[ \mathbf{\Lambda }_{k}\right] =\underset{-a}{%
\overset{a}{\dint }}\,\dfrac{1}{2a}e^{j\Delta \phi _{k,p}}\,d\Delta \phi _{k,p}=%
\frac{\sin \left( a\right) }{a} \text{ and}\notag \\
&\mathrm{E}\left[
\mathbf{\Xi }_{k}\right] =\underset{-b}{\overset{b}{\dint }}%
\dfrac{1}{2b}e^{-j\Delta \psi _{k,m}}\,d\Delta \psi _{k,m}=\frac{\sin \left(
b\right)}{b}, \label{Prop}\vspace*{-1mm}
\end{align}\vspace*{-0mm}%
respectively. Then, the received pilot symbols $\widehat{\mathbf{s}}_{k}^{T}$ used for the equivalent channel estimation at the $k$-th RF chain, which is under the impact of phase errors and transceiver analog beamforming errors, can be expressed as\vspace*{-0.0mm}%
\begin{equation}
\vspace*{-1mm}
\widehat{\mathbf{s}}_{k}^{T}=\widetilde{\mathbf{\gamma}}_{k}^{T}\mathbf{\Delta }%
_{k}\mathbf{\Xi }_{k}\overset{N}{\underset{i=1}{\dsum }}\mathbf{H%
}_{i}\mathbf{\Lambda }_{i}\mathbf{\Theta }_{i}\widetilde{\mathbf{\omega }}_{%
i}^{\ast }\sqrt{E_{\mathrm{P}}}\mathbf{\Phi }_{i}^{T}+\underset{\mathrm{effective}\text{ }\mathrm{noise}}{\underbrace{\widehat{\mathbf{z}}_{\mathrm{eq}}}}.\vspace*{-1mm}
\end{equation}\vspace*{-0.0mm}%
Similarly, we can express the average achievable rate per user as\vspace*{-0.0mm}%
\begin{equation}
\vspace*{-1.0mm}
\widehat{R}_{\mathrm{HB}}=\mathrm{E}_{\mathbf{H}_{\mathrm{S}}\mathrm{,\Delta
\phi ,\Delta \psi ,\Delta \theta }_{\mathrm{MS}}\mathrm{,\Delta \theta }_{%
\mathrm{BS}}}\left\{ \log _{2}\left[ 1+\frac{\widehat{\beta }^{2}E_{s}}{%
\sigma _{\mathrm{MS}}^{2}}\right] \right\} ,\vspace*{-1mm}
\end{equation}\vspace*{-0.00mm}%
where $\widehat{\beta}$ is the transmission power normalization factor under phase errors and transceiver beamforming errors.
Based on (\ref{HSUB_LA}), (\ref{HPBW_D}), and (\ref{Prop}), the approximation of the average achievable rate per user $\widehat{R}_{\mathrm{HB}}$ in the large numbers of antennas regime is given by\vspace*{-0mm}%
\begin{equation}
\vspace*{-2mm}
\hspace*{-1mm}
\widehat{R}_{\mathrm{HB}}\overset{M\rightarrow \infty }{\approx }\log
_{2}\left\{ 1+ \left[  \dfrac{\upsilon}{%
\upsilon+1}\frac{MP}{N}\widehat{\xi} +\dfrac{1}{\upsilon+1}%
 \right] \dfrac{E_{s}}{\sigma _{\mathrm{MS}}^{2}}\right\},
\label{QHI}\vspace*{-0mm}
\end{equation}\vspace*{-0.0mm}%
where $\widehat{\xi}=\left( \frac{\sin \left( a\right) }{a}\right) ^{2}\left( \frac{\sin \left( b\right) }{b}\right) ^{2}\xi$ is the equivalent power loss coefficient.
From Equation (\ref{QHI}), we note that the joint impact of random phase errors and transceiver beamforming errors on the average achievable rate cause performance degradation compared to the case of perfect hardware. Besides, using extra transmission power can compensate the performance degradation.

\vspace*{-0mm}
\subsection{Hardware Impairment and  Imperfect
Channel Estimation}
In this section, we further study the system performance by jointly taking account the impact of hardware impairments and equivalent channel estimation errors caused by noise.
The estimated equivalent channel $\widetilde{\mathbf{H}}_{\mathrm{eq}}^{T}$ can be expressed as\vspace*{-0mm}
\begin{equation}
\vspace*{-0.5mm}
\widetilde{\mathbf{H}}_{\mathrm{eq}}^{T}=\widehat{\mathbf{H}}_{\mathrm{eq}%
}^{T}+\Delta \widehat{\mathbf{H}}_{\mathrm{eq}}^{T}\mathbf{,}\vspace*{-0.5mm}
\end{equation}\vspace*{-0mm}%
where $\widehat{\mathbf{H}}_{\mathrm{eq}}$ is the equivalent channel under random phase errors and transceiver beamforming errors and the entries of the normalized channel estimation error $\frac{1}{\sqrt{MP}}\Delta \widehat{\mathbf{H}}_{\mathrm{eq}}^{T}$ are modeled by i.i.d. random variables with distribution $\mathcal{CN}\left( 0,\delta ^{2}\right)$, $\delta ^{2}=\mathrm{MSE}_{\mathrm{eq}}$, which is given by Equation (\ref{MSE_NN}).
The ZF precoding matrix for the considered hybrid system based on the imperfect CSI is given by\vspace*{-0mm}%
\begin{align}
\vspace*{-2mm}
\widetilde{\mathbf{W}}_{\mathrm{eq}}&=\widetilde{\mathbf{H}}_{\mathrm{eq}%
}^{\ast }\left[ \widetilde{\mathbf{H}}_{\mathrm{eq}}^{T}\widetilde{\mathbf{H}%
}_{\mathrm{eq}}^{\ast }\right] ^{-1}=\widehat{\mathbf{W}}_{\mathrm{eq}%
}+\Delta \widehat{\mathbf{W}}_{\mathrm{eq}}\notag \\
&=\left[
\begin{array}{ccc}
\widetilde{\mathbf{w}}_{\mathrm{eq,}1} &\cdots& \widetilde{\mathbf{w}}_{%
\mathrm{eq,}N}%
\end{array}%
\right] ,\vspace*{-1mm}
\end{align}\vspace*{-0mm}%
where $\widetilde{\mathbf{w}}_{\mathrm{eq,}k}\in\mathbb{C}^{N\times 1}$ is the $k$-th column vector of $\widetilde{\mathbf{W}}_{\mathrm{eq}}$ and $\widehat{\mathbf{W}}_{\mathrm{eq}}$ is the precoder based on $\widehat{\mathbf{H}}_{\mathrm{eq}}$. 
The received signal at user $k$ under imperfect CSI is given by\vspace*{-0mm}%
\begin{equation}
\vspace*{-1mm}
\widetilde{y}_{\mathrm{ZF}}^{k}=\underset{\mathrm{desired}\text{ }\mathrm{%
signal}}{\underbrace{\widetilde{\beta }x_{k}}}+\underset{\mathrm{intra-cell}%
\text{ }\mathrm{interference}}{\underbrace{\widetilde{\beta }\widehat{%
{\mathbf{\omega }}}_{k}^{H}\mathbf{H}_{k}^{T}\widehat{\mathbf{F}}_{%
\mathrm{RF}}\Delta \widehat{\mathbf{W}}_{\mathrm{eq}}\mathbf{x}}}+z_{k},\vspace*{-0mm}
\end{equation}\vspace*{-0mm}%
where $\widetilde{\beta }=\sqrt{\tfrac{1}{\mathrm{tr}\left( \widetilde{\mathbf{W}}_{\mathrm{eq}}\widetilde{\mathbf{W}}_{\mathrm{eq}}^{H}\right) }}$ is the power normalization factor. $\Delta \widehat{\mathbf{w}}_{\mathrm{eq,}j}\in\mathbb{C}^{M\times 1}$ denotes the $j$-th column vector of the ZF precoder error matrix $\Delta\widehat{\mathbf{W}}_{\mathrm{eq}}\mathbf{=}\widetilde{\mathbf{W}}_{\mathrm{eq}}-\widehat{\mathbf{W}}_{\mathrm{eq}}$ and $\mathbf{x}=[x_{1},\,x_{2},\ldots ,$ $x_{N}]^{T}$ denotes the transmitted signal for all users.
We then express the SINR of user $k$ as \vspace*{-0mm}
\begin{equation}
\hspace*{-2mm}
\widetilde{\mathrm{SINR}}_{\mathrm{ZF}}^{k}=\frac{E_{s}{\widetilde{\beta }^{2}}%
}{{\widetilde{\beta }^{2}}E_{s}\widehat{\mathbf{h}}_{\mathrm{eq,}k}^{T}\mathrm{%
E}_{\mathrm{\Delta}\widehat{\mathrm{\mathbf{H}}}_{\mathrm{eq}%
}}\left[ \Delta \widehat{\mathbf{W}}_{\mathrm{eq}}\Delta \widehat{\mathbf{W}}%
_{\mathrm{eq}}^{H}\right] \widehat{\mathbf{h}}_{\mathrm{eq,}k}^{\ast
}+{\sigma_{\mathrm{MS}}^{2}}}.  \label{SINR_with_error}
\end{equation}\vspace*{-0mm}%
Now we summarize the achievable rate per user in the high SNR regime in the following theorem.
\begin{theo}\vspace{-0.0mm}
As the receive $\mathrm{SNR=}\frac{E_{s}}{\sigma _{\mathrm{MS}}^{2}}$ approaches infinity, the approximated achievable rate of user $k$ of ZF precoding under imperfect hybrid CSI is given by\vspace*{-0mm}
\begin{align}
\hspace*{-2mm}
\widetilde{R}_{\mathrm{ZF},k}&\approx\log _{2} \left\{ 1+\left[ \left(
\sqrt{1+\delta ^{2}}-1\right) ^{2}\right. \right.\notag \\
&-2\sqrt{1+\delta ^{2}}\left( \sqrt{1+\delta ^{2}}-1\right) \delta ^{2}N\eta _{kk}\notag \\
& \left. \left. +\left( \sqrt{1+\delta ^{2}}\right) \left( 2-\sqrt{1+\delta
^{2}}\right) \delta ^{2}\mathrm{tr}\left( \mathbf{K}^{-1}\right) \right]
^{-1}\right\} ,  \label{Theo_IP_RL}\vspace*{-1mm}
\end{align}\vspace{-0.0mm}%
where $\mathbf{K}=\widehat{\mathbf{H}}_{\mathrm{eq}}^{T}\widehat{\mathbf{H}}%
_{\mathrm{eq}}^{\ast }$, $\eta _{kk}$ represents the $k$-th diagonal element
of $\mathbf{K}^{-1}$, and $\delta ^{2}$ is the normalized MSE of channel estimation given by Equation (\ref{MSE_NN}).
\end{theo}\vspace{-0mm}
\begin{proof}
Please refer to Appendix B.
\end{proof}
\begin{coro}\vspace{-0.0mm}%
In the large numbers of antennas regime, i.e., $M\rightarrow \infty$, %
the asymptotic average achievable rate per user of the hybrid system is approximated by\vspace*{-1mm}
\begin{align}
\hspace*{-2mm}
\widetilde{R}_{\mathrm{ZF}}& \approx \log _{2}\left\{ 1+\left[ \left(
\sqrt{1+\delta ^{2}}-1\right) ^{2}\right. \right.  \notag \\
&-2\sqrt{1+\delta ^{2}}\left( \sqrt{%
1+\delta ^{2}}-1\right)\frac{\delta ^{2}N}{\widehat{\xi} MP}\frac{%
\upsilon +1}{\upsilon } \notag \\
& \left. \left. +\left( \sqrt{1+\delta ^{2}}\right) \left( 2-\sqrt{1+\delta
^{2}}\right)\frac{\delta ^{2}N}{\widehat{\xi} MP}\frac{\upsilon +1}{%
\upsilon }\right] ^{-1}\right\} .  \label{Coro_RZFK}\vspace*{-3mm}
\end{align}
\end{coro}\vspace{-0.0mm}
\begin{proof}
The result follows by substituting $\mathbf{K}\underset{M\rightarrow \infty }{\overset{a.s.}{\rightarrow }}\widehat{\xi} MP \frac{\upsilon}{%
\upsilon+1}\mathbf{I}_N$ into (\ref{Theo_IP_RL}).
\end{proof}\vspace{+1mm}
Based on (\ref{QHI}) and (\ref{Coro_RZFK}), the additional achievable rate performance degradation, which is caused jointly by random phase errors and transceiver beamforming errors, is further summarized in the following corollary.
\begin{coro}\vspace{-0.0mm}%
In the large numbers of antennas regime, i.e., $M\rightarrow \infty$, the approximated achievable rate per user performance gap between the system with ideal hardware and the system under phase errors and transceiver beamforming errors is given by\vspace*{-0mm}
\begin{equation}
\vspace*{-0mm}
\Delta \mathrm{Gap}\approx \log _{2}\left[ \frac{1}{\widehat{\xi} }\right].  \label{Coro_GAPforHI}
\end{equation}
\end{coro}\vspace{-0.0mm}
\begin{proof}
The result comes after some mathematical manipulation on (\ref{QHI}) and (\ref{Coro_RZFK}).
\end{proof}

\vspace*{-0mm}
\section{Simulation and Discussion}
\begin{figure}[t]
\centering\vspace{-1mm}
\includegraphics[width=3.5in]{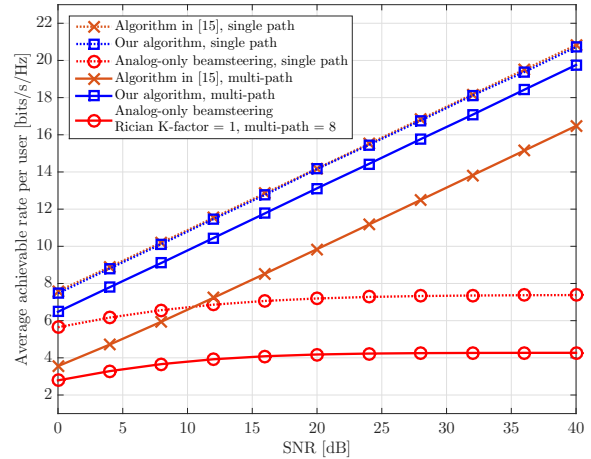}\vspace{-0mm}
\caption{The average achievable rate per user (bits/s/Hz) versus SNR for  algorithm proposed in \cite{Alkhateeb2015}
 and our proposed algorithm.}
\label{fig:Comp_HB_ALH_ours}\vspace{-1mm}
\end{figure}
In this section, we present further numerical results to validate our analysis.
We consider a small single-cell hybrid mmWave system.

In Figure \ref{fig:Comp_HB_ALH_ours}, we compare the achievable rates using the proposed algorithm and the algorithm proposed by \cite{Alkhateeb2015} for sparse and non-sparse mmWave channels.
We assume perfect channel estimation with $M=100,$ $N=4,$ and $P=16$. For non-sparse mmWave channels, $\upsilon_{k} =1,\forall k$.
Firstly, we illustrate the effectiveness of the proposed channel estimation algorithm over non-sparse mmWave channels.
For sparse single-path channels, the achievable rate of the proposed algorithm matches with the algorithm proposed in \cite{Alkhateeb2015}.
For non-sparse mmWave channels, e.g. with the number of multi-paths $N_{\mathrm{l}}=8$, we observe that the proposed algorithm achieves a better system performance than that of the algorithm proposed in \cite{Alkhateeb2015}.
The reasons are two-fold. Firstly, the proposed algorithm can effectively reduce the MU interference via the proposed analog beamformers and digital ZF precoder.
In particular, the analog beamformers adopted at the BS and the users allow transmission and receiving along the strongest AoA directions, respectively.
Therefore, the interference outside the strongest AoA directions is reduced.
Secondly, the digital ZF precoder designed based on the equivalent channels can remove the MU interference within the strongest AoAs.
We note that since the analog beamforming only allows the transmission along the strongest AoA directions, the proposed scheme can utilize the transmission power more efficiently.
In contrast, the algorithm proposed in \cite{Alkhateeb2015}, which aims to maximize the desired signal energy, does not suppress the MU interference as effective as our proposed algorithm.
Furthermore, Figure \ref{fig:Comp_HB_ALH_ours} also illustrates that a significant achievable rate gain can be brought by the proposed channel estimation and the digital ZF precoding over a simple analog-only beamforming steering scheme.

In Figure \ref{fig:HWI_CSIerror}, we illustrate two sets of comparisons to validate our derived results in Equations (\ref{QHI}) and (\ref{Coro_RZFK}), and to show the impact of hardware impairments on system performance.
In the simulations, we set the Rician K-factor as $\upsilon_{k} =2, \forall k$, the number of antennas equipped at the BS is $M=100$, the number of antennas equipped at each user is $P=8$, and the number of users is $N=8$.
We assume that the maximum phase error value of phase shifters is $\sigma _{\Delta }=3$ degrees and the variance of AoA estimation errors at the BS side is $\varrho _{\mathrm{BS}}^{2}= \frac{1.782}{2M}$.
In this case, the equivalent coefficient is $\widehat{\xi} \approx 0.5$, which predicts that there is a $3$-dB loss in SNR caused by AoA estimation errors and phase errors.
First, we observe that even with transceiver beamforming errors and random phase errors, the achievable rate of the proposed scheme (curve $2$) scales with increasing SNR and is unbounded by above.
The small gap between the achievable rate per user with perfect hardware (curve $1$) and the achievable rate per user with hardware impairment (curve $2$) confirms the robustness of the proposed scheme against random phase errors as well as transceiver beamforming errors, which is predicted by Equation (\ref{QHI}).
In addition, the $3$ dB of extra power consumption caused by hardware impairments is also verified via the comparison between curves $1$ and $2$.

\begin{figure}[t]
\centering\vspace{-0mm}
\includegraphics[width=3.5in,]{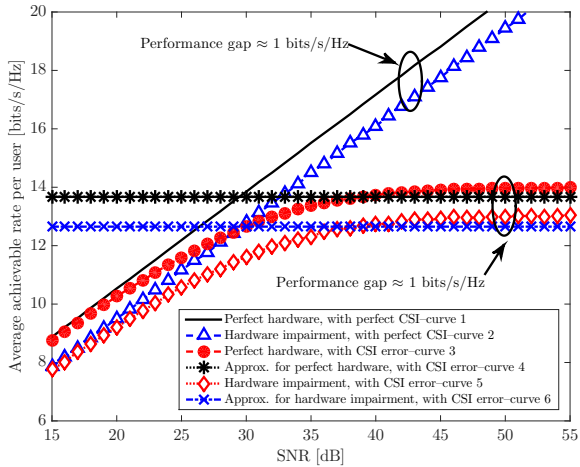}\vspace{-0mm}
\caption{The average achievable rate per user (bits/s/Hz) versus SNR for the hybrid system.}
\label{fig:HWI_CSIerror}\vspace{-1mm}
\end{figure}

Figure \ref{fig:HWI_CSIerror} also illustrates a comparison between the proposed scheme with imperfect equivalent CSI under perfect hardware (curve
$3$) and the proposed scheme with imperfect CSI under random phase errors and transceiver beamforming errors (curve $5$), which validates the correctness of Theorem $2$.
These two simulation curves assume an identical normalized CSI error variance $\delta^{2}$.
In particular, we set the normalized MSE of CSI as $\delta^{2}=0.005$, which is achieved by the proposed channel estimation scheme as shown in Figure \ref{fig:antiPC}.
In the high receive SNR and large numbers of antennas regimes, curve $3$ and curve $5$ can be asymptotically approximated by curve $4$ and curve $6$, respectively.
Interestingly, the performance gaps between curve $3$ and curve $5$ as well as curve $1$ and curve $2$ are approximately $\Delta \mathrm{Gap}\approx \log _{2}\left( \frac{1}{\widehat{\xi} }\right)= 1$ bits/s/Hz, which are accurately characterized by Equation (\ref{Coro_GAPforHI}).

Now, we evaluate the achievable rate versus the number of antennas in Figure \ref{fig:Comp_versus_No_antennas}.
The setup in Figure \ref{Fig004} is considered at receive $\mathrm{SNR}=30$ dB and the number of antennas equipped at each user is $P=8$, with different numbers of BS antennas.
Figure \ref{Fig004} shows that the achievable rate performance of the proposed digital ZF precoding increases with increasing numbers of antennas at the BS, despite the existence of CSI estimation errors.
This is mainly due to the fact that equipping more antennas at the BS and the users can lead to higher array gains.
It can be observed that the achievable rate of the proposed system scales with the number of BS antennas with a similar slope as the fully digital system, which shows the effectiveness of the proposed scheme for the hybrid system in exploiting the spatial degrees of freedom.
In Figure \ref{Fig005}, the SNR setup and $M=100$ are the same as Figure \ref{Fig004} but with different numbers of users antennas.
Similar phenomena are observed as in Figure \ref{Fig004}.
With the same value of $MP$ and identical normalized MSE $\delta^{2}$, it is shown in Figure \ref{Fig004} and Figure \ref{Fig005} that the achievable rate performance degradations (the gap between the dash lines and the solid lines) caused by CSI errors are identical.

\begin{figure}[t]
\centering\vspace{-0mm}
\begin{subfigure}{.5\textwidth}
  \includegraphics[width=3.5in]{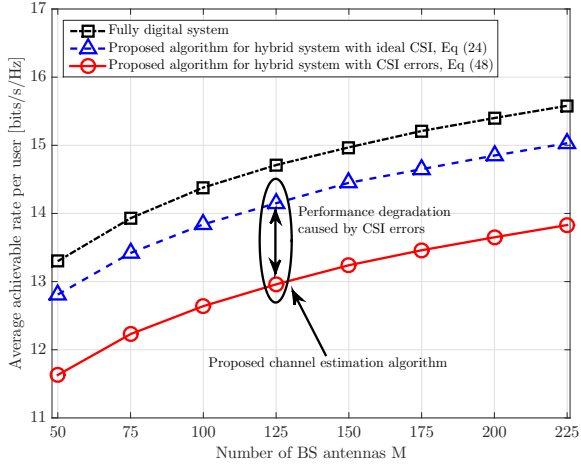}
  \caption{Average achievable rate versus BS antennas.}
  \label{Fig004}
  \end{subfigure}
\begin{subfigure}{.5\textwidth}
  \includegraphics[width=3.5in]{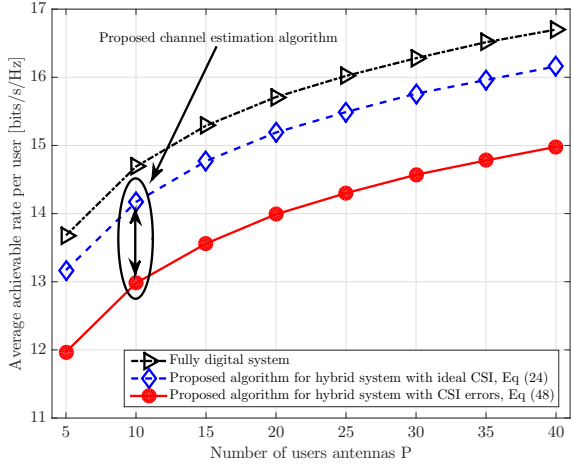}
  \caption{Average achievable rate versus users antennas.}
  \label{Fig005}
\end{subfigure}
\vspace*{-0mm}
\caption{(a) Average achievable rate (bits/s/Hz) versus different number of BS
antennas for different structures with $P=8$; (b) Average achievable
rate (bits/s/Hz) versus different number of users antennas for different structures with $M=100$. We set the number of users $N=8$, the
normalized CSI $\mathrm{MSE}=0.005$, $\mathrm{SNR}=30$ dB, and the Rician
K-factor of $2$ for both (a) and (b).}
\label{fig:Comp_versus_No_antennas}\vspace{-1mm}
\end{figure}

\vspace{+2mm}
\section{Conclusions}

In this paper, we proposed a low-complexity mmWave channel estimation algorithm exploiting the strongest AoA for the MU hybrid mmWave systems, which is applicable for both sparse and non-sparse mmWave channel environments.
The MSE performance of the proposed channel estimation was analyzed and verified via numerical simulation.
The achievable rate performance of designed analog beamforming and digital ZF precoding based on the proposed channel estimation scheme was derived and compared to that of the fully digital system.
The analytical and simulation results indicated that the proposed scheme can approach the rate performance achieved by the fully digital system with sufficiently large Rician K-factors.
By taking into consideration of the effects of random phase errors, transceiver beamforming errors, and CSI errors in the rate performance analysis, closed-form approximation of achievable rate in the high SNR regime was derived and verified via simulation.
Our results showed that the proposed scheme is robust against random phase errors and transceiver beamforming errors.
\vspace{-0mm}
\section*{Appendix}
\vspace{-0mm}
\subsection{Proof of Theorem 1}

\begin{equation}
\hspace{-2mm}
R_{\mathrm{HB}} =\mathrm{E}_{%
\mathbf{H}_{\mathrm{S}}}\left\{ \log _{2}\left[ 1+\left[ \mathrm{tr}\left[ (%
\mathbf{H}_{\mathrm{eq}}^{T}\mathbf{H}_{\mathrm{eq}}^{\ast })^{-1}\right] %
\right] ^{-1}\dfrac{E_{s}}{\sigma _{\mathrm{MS}}^{2}}\right] \right\} .
\label{Proof_1}
\end{equation}\vspace*{-0.0mm}%
First, we introduce some preliminaries. Since $\mathbf{H}_{\mathrm{eq}}^{T}\mathbf{H}_{\mathrm{eq}}^{\ast}$ is a positive definite Hermitian matrix, by eigenvalue decomposition, it can be decomposed as $\mathbf{H}_{\mathrm{eq}}^{T}\mathbf{H}_{\mathrm{eq}}^{\ast }=\mathbf{U\Lambda V}^{H}$, $\mathbf{\Lambda}\in\mathbb{C}^{N\times N}$ is the positive diagonal eigenvalue matrix, while $\mathbf{V}\in\mathbb{C}^{N\times N}$ and $\mathbf{U}\in\mathbb{C}^{N\times N}$ are unitary matrixes, $\mathbf{U=V}^{H}$.
The trace of the eigenvalues of $\mathbf{H}_{\mathrm{eq}}^{T}\mathbf{H}_{\mathrm{eq}}^{\ast}$ equals to the trace of matrix $\mathbf{\Lambda }$.
Then we can rewrite the power normalization factor in (\ref{Proof_1}) as\vspace*{-1.0mm}%
\begin{align}
\vspace*{-1mm}
\dfrac{N}{\mathrm{tr}\left[ (\mathbf{H}_{\mathrm{eq}}^{T}\mathbf{H}_{\mathrm{eq}}^{\ast })^{-1}\right] }&=N\left[ \mathrm{tr}\left[ \mathbf{U\Lambda U}^{H}\right] ^{-1}\right]^{-1}\label{Proof_2}\\
&=N\left[ \mathrm{tr}\left( \mathbf{\Lambda }^{-1}\right) \right] ^{-1}=\left[ \overset{N}{\underset{i=1}{\dsum }}\frac{1}{N}\lambda _{i}^{-1}\right]^{-1}, \vspace*{-1mm} \notag
\end{align}\vspace{-0.0mm}%
In addition, $f(x)=x^{-1}$, $x > 0$,  is a strictly decreasing convex function and exploiting the convexity, we have the following results \cite{book:infotheory}\vspace{-0mm}%
\begin{equation}
\vspace{-1.0mm}
\left[ \overset{N}{\underset{i=1}{\dsum }}\frac{1}{N}\lambda _{i}^{-1}\right]
^{-1}\leqslant \overset{N}{\underset{i=1}{\dsum }}\frac{1}{N}\left[ \left(
\lambda _{i}^{-1}\right) ^{-1}\right] =\overset{N}{\underset{i=1}{\dsum }}%
\frac{1}{N}\lambda _{i}.  \label{JIE}\vspace{-0.0mm}
\end{equation}\vspace{-0.0mm}%
Therefore, based on (\ref{Proof_2}) and (\ref{JIE}), we have the following inequality: \vspace{-0mm}%
\begin{align}
\vspace{-2.0mm}
\dfrac{1}{\mathrm{tr}\left[ \left( \mathbf{H}_{\mathrm{eq}}^{T}\mathbf{H}_{%
\mathrm{eq}}^{\ast }\right) ^{-1}\right] }\leqslant& \overset{N}{\underset{i=1%
}{\dsum }}\frac{1}{N^{2}}\lambda _{i}=\frac{1}{N^{2}}\overset{N}{\underset{%
i=1}{\dsum }}\lambda _{i}\notag\\
=&\frac{1}{N^{2}}\mathrm{tr}\left[ \mathbf{H}_{\mathrm{eq}}^{T}\mathbf{H}_{\mathrm{eq}}^{\ast }\right] .  \label{Proof_3} \vspace{-1.0mm}
\end{align}\vspace{-0.0mm}%
Based on (\ref{Proof_3}), Equation (\ref{Proof_1}) can be rewritten as\vspace{-1.0mm}%
\begin{align}
R_{\mathrm{HB}}\overset{(a)}{\leqslant }&\mathrm{E}_{\mathbf{H}_{\mathrm{S}%
}}\left\{ \log _{2}\left[ 1+\frac{1}{N^{2}}\mathrm{tr}\left[ \mathbf{H}_{%
\mathrm{eq}}^{T}\mathbf{H}_{\mathrm{eq}}^{\ast }\right] \dfrac{E_{s}}{\sigma
_{\mathrm{MS}}^{2}}\right] \right\} \notag\\
\overset{(b)}{\leqslant }&\log _{2}\left\{
1+\frac{1}{N^{2}}\mathrm{E}_{\mathbf{H}_{\mathrm{S}}}\left[ \mathrm{tr}%
\left( \mathbf{H}_{\mathrm{eq}}^{T}\mathbf{H}_{\mathrm{eq}}^{\ast }\right) %
\right] \dfrac{E_{s}}{\sigma _{\mathrm{MS}}^{2}}\right\}  \notag \\
=&\log _{2}\left\{ 1+\frac{1}{N^{2}}\dfrac{E_{s}}{\sigma _{\mathrm{MS}}^{2}}%
\left[ \left( \dfrac{\upsilon}{\upsilon+1}\right) MP\| \mathbf{F}_{\mathrm{RF}}^{H}\mathbf{F}_{\mathrm{RF}} \|_{\mathrm{F}}^{2} \right] \right. \notag \\
&\left. +\left( \dfrac{1}{%
\upsilon+1}\right)\dfrac{E_{s}}{\sigma _{\mathrm{MS}}^{2}} \right\}.
\label{Proof_HBUB}
\end{align}\vspace{-0.0mm}%
In $(a),$ we follow (\ref{Proof_3}) and in $(b),$ we adopt the Jensen's inequality. This completes the proof.

\vspace{-0mm}
\subsection{Proof of Theorem 2}
\vspace{-0mm}
The receive SINR of user $k$ is given by \vspace*{-0.0mm}%
\begin{equation}
\vspace*{-1mm}
\widetilde{\mathrm{SINR}}_{\mathrm{ZF}}^{k}=\frac{\widetilde{\beta }^{2}E_{s}%
}{\underset{\mathrm{Interference}\text{\ }\mathrm{term}\text{\ }\mathrm{due}\text{\ }\mathrm{to}\text{\ }\mathrm{errors}}{\underbrace{\widetilde{\beta }^{2}E_{s}\widehat{\mathbf{h}}_{%
\mathrm{eq,}k}^{T}\mathrm{E}_{{\Delta}\widehat{\mathbf{H}}_{\mathrm{eq}}}\left[ \Delta \widehat{\mathbf{W}}_{\mathrm{eq}%
}\Delta \widehat{\mathbf{W}}_{\mathrm{eq}}^{H}\right] \widehat{\mathbf{h}}_{%
\mathrm{eq,}k}^{\ast }}}+\sigma _{\mathrm{MS}}^{2}}.\vspace*{-0mm}
\end{equation}\vspace{-0.0mm}%
To express $\mathrm{E}_{{\Delta}\widehat{\mathbf{H}%
}_{\mathrm{eq}}}\left[ \Delta \widehat{\mathbf{W}}_{\mathrm{eq}}\Delta
\widehat{\mathbf{W}}_{\mathrm{eq}}^{H}\right]$, first we present the
normalized expression of\textbf{\ }$\Delta \widehat{\mathbf{W}}_{\mathrm{eq}%
} =$\vspace*{-1.0mm}
\begin{align}
\hspace*{-2.0mm}
\sqrt{1+\delta ^{2}}&(\widehat{\mathbf{H}}_{\mathrm{eq}}^{\ast}+\Delta \widehat{%
\mathbf{H}}_{\mathrm{eq}}^{\ast})\left[ (\widehat{\mathbf{H}}_{\mathrm{eq}%
}+\Delta \widehat{\mathbf{H}}_{\mathrm{eq}})^{T}(\widehat{\mathbf{H}}_{%
\mathrm{eq}}^{\ast}+\Delta \widehat{\mathbf{H}}_{\mathrm{eq}}^{\ast})\right] ^{-1}\notag\\%
&-\widehat{\mathbf{H}}_{\mathrm{eq}}^{\ast }\left( \widehat{\mathbf{H}}_{%
\mathrm{eq}}^{T}\widehat{\mathbf{H}}_{\mathrm{eq}}\right)^{-1},
\label{Eq_222}\vspace*{-1.0mm}
\end{align}%
where $\mathbf{K}=(\widehat{\mathbf{H}}_{\mathrm{eq}}^{T}\widehat{\mathbf{H}}%
_{\mathrm{eq}}^{\ast })$, \vspace*{+1mm}$\mathbf{D}=(\Delta \widehat{\mathbf{H}}_{\mathrm{%
eq}}^{T}\widehat{\mathbf{H}}_{\mathrm{eq}}^{\ast }+\widehat{\mathbf{H}}_{%
\mathrm{eq}}^{T}\Delta \widehat{\mathbf{H}}_{\mathrm{eq}}^{\ast }+\Delta
\widehat{\mathbf{H}}_{\mathrm{eq}}^{T}\Delta \widehat{\mathbf{H}}_{\mathrm{eq%
}}^{\ast }),$\vspace*{+1mm} and $\widehat{\mathbf{W}}_{\mathrm{eq}}=\widehat{%
\mathbf{H}}_{\mathrm{eq}}^{\ast }\mathbf{K}^{-1}.$ The matrix inversion
approximation is given by\vspace*{-0.0mm}
\begin{equation}
\vspace*{-1mm}
\mathbf{(K+D)}^{-1}\approx \left[ \mathbf{K}^{-1}-\mathbf{K}^{-1}\mathbf{DK}%
^{-1}\right] .\vspace*{-1mm}
\end{equation}\vspace*{-0.0mm}%
In this case, we re-express (\ref{Eq_222}) as\vspace*{-0.0mm}
\begin{align}
\vspace*{-1mm}
\Delta \widehat{\mathbf{W}}_{\mathrm{eq}}\approx &\sqrt{1+\delta ^{2}}\left(
\widehat{\mathbf{H}}_{\mathrm{eq}}^{\ast }+\Delta \widehat{\mathbf{H}}_{%
\mathrm{eq}}^{\ast }\right) \left( \mathbf{K}^{-1}-\mathbf{K}^{-1}\mathbf{DK}^{-1}\right) \notag \\
&-\widehat{\mathbf{H}}_{\mathrm{eq}}^{\ast }\mathbf{K}^{-1}.\vspace*{-1mm}
\end{align}\vspace*{-0mm}%
Finally, we have $E_{\mathrm{\Delta}\widehat{\mathrm{\mathbf{H}}}_{\mathrm{eq}%
}}\left[ \Delta \widehat{\mathbf{W}}_{\mathrm{eq}}\Delta \widehat{\mathbf{W}}%
_{\mathrm{eq}}^{H}\right] $
\begin{align}
\hspace{-2.0mm}
&\overset{\left( b\right) }{\approx } (\sqrt{1+\delta ^{2}}-1)^{2}\left[
\widehat{\mathbf{W}}_{\mathrm{eq}}\widehat{\mathbf{W}}_{\mathrm{eq}}^{H}%
\right] \notag\\
&+2\sqrt{1+\delta ^{2}}(1-\sqrt{1+\delta ^{2}})\delta ^{2}N\left[ \widehat{\mathbf{W}}_{\mathrm{eq}}\mathbf{K}^{-1}\widehat{\mathbf{%
W}}_{\mathrm{eq}}^{H}\right]  \notag \\
& +(2\sqrt{1+\delta ^{2}}-\sqrt{1+\delta ^{2}})\delta ^{2}\left[ \mathrm{tr}(%
\mathbf{K}^{-1})\widehat{\mathbf{W}}_{\mathrm{eq}}\widehat{\mathbf{W}}_{%
\mathrm{eq}}^{H}\right]  \label{Eq_30_0} \\
&\overset{(d)}{\approx } (\sqrt{1+\delta ^{2}}-1)^{2}\left[ \widehat{\mathbf{W}}_{%
\mathrm{eq}}\widehat{\mathbf{W}}_{\mathrm{eq}}^{H}\right]\notag\\
&+\frac{2N}{\xi MP}\sqrt{1+\delta ^{2}}(1-\sqrt{1+\delta ^{2}})\delta ^{2}\widehat{%
\mathbf{W}}_{\mathrm{eq}}\mathbf{G}_{\mathrm{L}}^{-1}\widehat{\mathbf{W}}_{%
\mathrm{eq}}^{H}  \notag \\
& +(2\sqrt{1+\delta ^{2}}-\sqrt{1+\delta ^{2}})\delta ^{2}\left[ \frac{N}{\xi MP}\widehat{\mathbf{W}}_{\mathrm{eq}}\mathbf{G}_{\mathrm{%
L}}^{-1}\widehat{\mathbf{W}}_{\mathrm{eq}}^{H}\right].  \label{Eq_30}\vspace{-1.0mm}
\end{align}\vspace{-0.0mm}%
In (b), we omit some negligibly small parts which neither dominate the performance nor scale with $M$.
In (d), while the number of antennas $M\rightarrow \infty $, $\mathbf{K}=\widehat{\mathbf{H}}_{\mathrm{eq}}^{T}\widehat{\mathbf{H}}_{\mathrm{eq}}^{\ast }\underset{M\rightarrow \infty }{\overset{a.s.}{\approx }} \xi MP\mathbf{G}_{\mathrm{L}}$ holds, where $\xi \in \left( 0,1\right]$.
We substitute (\ref{Eq_30_0}) into (\ref{SINR_with_error}) and substitute (\ref{Eq_30}) into (\ref{SINR_with_error}), the expressions (\ref{Theo_IP_RL}) and (\ref{Coro_RZFK}) come immediately after some straight forward mathematical manipulation.\vspace*{-0mm}


\end{document}